\documentclass[a4paper]{article}     

\usepackage{acronym}
\usepackage[linesnumbered,titlenumbered,ruled]{algorithm2e}
\usepackage{amsfonts}
\usepackage{amsmath}
\usepackage{amssymb}
\usepackage{amsthm}
\usepackage{authblk}
\usepackage[american]{babel}
\usepackage{color}
\usepackage{dsfont}
\usepackage{enumitem}
\usepackage{fixltx2e}
\usepackage{float}
\usepackage[T1]{fontenc}
\usepackage{gensymb}
\usepackage[left=30mm,right=30mm,top=30mm,bottom=30mm]{geometry}
\usepackage{graphicx}
\usepackage[utf8x]{inputenc}
\usepackage{subfigure}
\usepackage{url}
\usepackage{wrapfig}
\usepackage{xspace}
\usepackage[all]{xy}
\usepackage[pdfborder={0 0 0}]{hyperref}

\renewcommand{\epsilon}{\varepsilon}
\renewcommand{\rho}{\varrho}

\newcommand{\AG}{\operatorname{AGP}}

\newcommand{\AGC}{\operatorname{AGPF}}
\newcommand{\AGCrho}{\operatorname{AGPF\rho}}

\newcommand{\AGCstep}{\operatorname{\AGC\tau}}
\newcommand{\altrho}{\tilde{\rho}}
\newcommand{\argmin}[1]{\underset{#1}{\arg\min}\;}
\newcommand{\bigO}[1]{\operatorname{O}(#1)}
\newcommand{\calA}{\mathcal{A}}
\newcommand{\calS}{\mathcal{S}}
\def\dash---{\kern.16667em---\penalty\exhyphenpenalty\hskip.16667em\relax}
\newcommand\diam{\operatorname{diam}}
\newcommand{\N}{\mathds{N}}
\newcommand{\OPT}{\operatorname{OPT}}
\newcommand{\R}{\mathds{R}}
\newcommand{\st}{\textnormal{s.\,t.}}
\newcommand{\V}{\operatorname{\mathcal{V}}}
\newcommand{\vis}[1]{\V(#1)}
\newcommand{\Z}{\mathds{Z}}

\def\hyperspc{\kern -0.25em}
\newcommand{\lhfloor}{\lfloor \hyperspc \lfloor}
\newcommand{\rhfloor}{\rfloor \hyperspc \rfloor}

\newcommand{\Talpha}{Falloff~$\alpha$}
\newcommand{\Tradius}{Scaling~$\Lambda$}
\newcommand{\Tcontinuous}{\textsc{Continuous}\xspace}
\newcommand{\Tdiscrete}{\textsc{Discrete}\xspace}
\newcommand{\Tdiscretecircle}{\textsc{Discrete\-Circle}\xspace}
\newcommand{\Tdiscreteoctagon}{\textsc{Discrete\-Octagon}\xspace}

\newcommand{\orthoortho}{\textsc{Non\-Simple\-Ortho}\xspace}
\newcommand{\spike}{\textsc{Spike}\xspace}
\newcommand{\vonkoch}{\textsc{Von\-Koch}\xspace}
\newcommand{\simple}{\textsc{Simple}\xspace}
\newcommand{\simplesimple}{\textsc{Non\-Simple}\xspace}

\newtheorem{definition}{Definition}
\newtheorem{theorem}[definition]{Theorem}
\newtheorem{lemma}[definition]{Lemma}

\makeatletter\begin{document}

\title{\LARGE\bf Algorithms for Art Gallery Illumination}

\author[1]{Maximilian Ernestus}
\author[2,3]{Stephan Friedrichs}
\author[1]{Michael Hemmer}
\author[1]{Jan Kokem\"uller}
\author[1]{Alexander Kr\"oller}
\author[4]{Mahdi Moeini}
\author[5]{Christiane Schmidt}

\affil[1]{TU Braunschweig, IBR, Algorithms Group\newline M\"uhlenpfordtstr.\ 23, 38106 Braunschweig, Germany\newline \texttt{maximilian@ernestus.de}, \texttt{mhsaar@gmail.com},\newline \texttt{jan.kokemueller@gmail.com}, \texttt{kroeller@perror.de}}

\affil[2]{Max Planck Institute for Informatics, Saarbr\"ucken, Germany\newline \texttt{sfriedri@mpi-inf.mpg.de}}
\affil[3]{Saarbr\"ucken Graduate School of Computer Science}

\affil[4]{Chair of Business Information Systems and Operations Research (BISOR)\newline Technical University of Kaiserslautern\newline Postfach 3049, Erwin-Schr\"{o}dinger-Str., D-67653 Kaiserslautern, Germany.\newline \texttt{mahdi.moeini@wiwi.uni-kl.de}}

\affil[5]{Communications and Transport Systems, ITN, Link\"oping University, Sweden\newline \texttt{christiane.schmidt@liu.se}}

\date{}

\maketitle

\begin{abstract}
The Art Gallery Problem~(AGP) is one of the classical problems in
computational geometry. It asks for the minimum number of guards
required to achieve visibility coverage of a given polygon. The AGP
is well-known to be NP-hard even in restricted cases. In this
paper, we consider the Art Gallery Problem with Fading~(AGPF): A polygonal region
is to be illuminated with light sources such that every point is
illuminated with at least a global threshold, light intensity decreases
over distance, and we seek to minimize the total energy consumption.
Choosing fading exponents of zero, one, and two are equivalent to the
AGP, laser scanner applications, and natural light, respectively.
We present complexity results as well as a negative solvability result.
Still, we propose two practical algorithms for AGPF
with fixed light positions (e.g.\ vertex guards) independent of
the fading exponent, which we demonstrate to work well in practice. One is based on a discrete approximation, the other
on non-linear programming by means of simplex-partitioning
strategies. The former approach yields a fully polynomial-time
approximation scheme for AGPF with fixed light positions.
The latter approach obtains better results in our experimental
evaluation.
\end{abstract}

\section{Introduction}
\label{sec:intro}

The classical Art Gallery Problem~(AGP) asks for the minimum
number of guards required to cover boundary and interior of a polygon.
This is one of the best-known problems in computational
geometry; see the excellent book by O'Rourke~\cite{r-agta-87} for
an introduction to the subject.
In the classical AGP guards have an infinite range of visibility. We consider
a problem variant where light fades over distance.
The main applications are discussed below.

\subsection{Laser Scanning}

Consider obtaining an exact, two-dimensional laser scan of some indoor environment.
A typical $360^\circ$ laser scanner can be placed at any point. It operates by
taking a sample in one direction then turning by a configurable
angle $\theta$ and repeating the process. The time $t$ to obtain a scan is roughly $t\sim
\theta^{-1}$ and can range anywhere between seconds and several hours in
real-world applications. The quality of the result depends on the sample-point
density $q$ on the walls. For an object at distance $d$ this
is about $q\sim d^{-1}$. Together, we get $q\sim
td^{-1}$. The actual fading of the laser-light intensity is
irrelevant in practice. The problem of minimizing the time spent
scanning (with one scanner) while maintaining sufficient scanning
quality is an AGP generalization with linear fading.

\subsection{Realistic Light}

Light suffers from quadratic fading over distance.
From a dimmable light bulb at brightness $b$ an object at distance $d$ receives light in an amount of $q \sim bd^{-2}$.
We assume a linear correspondence between energy consumption and brightness.
The AGP generalization of illuminating a polygonal area with the minimum total energy was introduced by O'Rourke in 2005~\cite{dr-opcccg-06}. Since then only a restricted case of this problem has been addressed by Eisenbrand et~al.~\cite{efkm-esi-05}.

\subsection{Our Contribution}

We study two versions of the Art Gallery Problem with Fading~(AGPF).
One is the discrete version where guards can only be placed on finitely many a-priori known locations (e.g.\ vertex guards). The other version allows guards that can be positioned anywhere in the input polygon (point guards).
We address the question of computational complexity as well as solvability of AGPF (both the discrete and the general variant).

Two algorithms for AGPF with fixed guard positions constitute the core of our contribution.
That is, for the case that the selected guard locations have to be chosen from a given, finite candidate-location set.
Our algorithms are based on infinite LP formulations of the problem.
Both solve the problem of finding the darkest point in an intermediate solution, referred to as the Primal Separation Problem~(PSP), in different ways.
Our algorithms work for arbitrary fading exponents; they can be applied to the laser scanning as well as to the natural light scenario.

The first algorithm approximates the PSP by replacing continuous fading with an appropriate step function.
For a finite guard candidate set, such as vertex guards, we illustrate that it yields a Fully Polynomial-Time Approximation Scheme~(FPTAS).

Our second algorithm solves the PSP by means of a derivative-free approach based on triangulation of the polygon's visibility coverage, followed by optimizing non-linear programs by using a simplex-partitioning procedure~\cite{go-action-1996}.

\subsection{Overview}
Related work is presented in Section~\ref{sec:rw}.
Sections~\ref{sec:prel} and~\ref{sec:comp} formally introduce the problem, and discuss fading functions and complexity results.
In Section~\ref{sec:fixedguards} we turn our attention to the case of fixed, finite sets of guard positions (e.g.\ vertex guards) and propose two algorithms, one of which yields a FPTAS.
We evaluate our algorithms experimentally in Section~\ref{sec:exp}.
Section~\ref{sec:point} gives a brief introduction to the general (point guard) case where guards can be placed anywhere in the polygon.
We conclude this paper in Section~\ref{sec:concl}.

\section{Related Work}\label{sec:rw}

Guarding problems have been studied for several decades.
Chv\'atal~\cite{c-actpg-75} was the first to answer Victor Klee's
question on the number of guards that are always sufficient and
sometimes necessary to monitor a polygon $P$ with $n$
vertices; $\left\lfloor\frac{n}{3}\right\rfloor$ is the tight bound for general
polygons which was soon afterwards shown by a beautiful proof of
Fisk~\cite{f-spcwt-78}. Many other polygon classes have been
considered w.r.t.\ this kind of Art Gallery Theorems. For example, Kahn
et~al.~\cite{kkk-tgrfw-83} obtained a tight bound of
$\left\lfloor\frac{n}{4}\right\rfloor$ for orthogonal polygons. Generally, the guards are allowed to be located
anywhere in the polygon (point guards). Sometimes,
less powerful guards have been considered, such as guards that must be
placed on vertices of $P$ (vertex guards), or guards that only have a certain illumination angle
 $\gamma$. For the latter problem, T\'oth~\cite{t-agpgr-00} showed an
 art gallery theorem with tight bound $\left\lfloor\frac{n}{3}\right\rfloor$ for $\gamma=180^\circ$.
 While the above results give bounds on the number of guards in certain classes of polygons, the AGP asks for the minimum number of guards that cover a given polygon $P$.
 Several versions of this problem were shown to be NP-hard, for example the problem restricted to guards placed on vertices of simple polygons, as shown by Lee and Lin~\cite{ll-ccagp-86}.
See~\cite{rsfhkt-eag-14,r-agta-87} for surveys on the AGP.

The Art Gallery Problem with Fading~(AGPF) has only been considered for the case of fixed light-source positions in the plane and the illumination of a  line (the stage).
Eisenbrand et~al.~\cite{efkm-esi-05} aim at minimizing the total amount of power assigned to the given light sources such that the entire stage is lit.
The authors give a convex programming formulation of the problem that leads to a polynomial-time solution. They present a $(1+\epsilon)$- as well as an $\bigO{1}$-approximation.

Over the last few years, efficient implementations for the classical AGP emerged~\cite{crs-ipsagp-09-XXXX-STATTDESSEN-DER-TR-XXX,csr-eeaoagp-07,ffks-ffagp-14,kbfs-esbgagp-12,kms-aneafstagp-13,DaviPedroCid-J-002013}.
Currently, the state of the art is that optimal solutions for polygons with several thousand of vertices can be found efficiently. See~\cite{rsfhkt-eag-14} for a survey on that matter.

\section{Notation and Preliminaries}
\label{sec:prel}

We consider a given polygon $P$, possibly with holes.
$P$ has vertices $V$, with $|V| = n$.
The diameter of $P$ is denoted $D := \diam P$.
For a point $p\in P$ the {\em visibility region} $\vis{p}$ is the (star-shaped) set of all points of $P$ visible from $p$.
In the original AGP, we say a {\em guard set} $G \subseteq P$ {\em covers} $P$ if and only if $\bigcup_{g\in G} \vis{g} = P$.
The goal is to find a covering $G$ of minimum cardinality.

In this paper, we consider a variant of this problem: We assume the
guards $g\in G$ to be light sources whose intensity (and
therefore energy consumption) can be controlled. We denote the
intensity of $g\in G$ by $0 \leq x_g \in \R$. Furthermore, light suffers from
fading over distance, modeled by a fading function $\rho$:
While a point $w \not\in \vis{g}$ does not receive any light from $g$, $w \in \vis{g}$ receives $\rho(g,w)x_g$, i.e., $\rho$ accounts for the possible fading exponents motivated in Section~\ref{sec:intro}.

For a set $W\subseteq P$ of {\em witnesses} \/(points that witness sufficient or insufficient illumination), we say a guard set
$G\subseteq P$ with intensities $x \in \R^G_{\geq 0}$ {\em illuminates} $W$ if and only if every point in
$W$ is sufficiently lit. Without loss of generality, we normalize the meaning
of $w \in W$ being \emph{sufficiently lit} to
\begin{equation}
  \sum_{g \in G \cap \vis{w}} \rho(g,w)x_g \geq 1\;.
 \label{eq:fadingineq}
\end{equation}
Given a set of possible light sources~$G$, a set $W$ of witness points, and a fading function~$\rho$, our objective is to illuminate $W$ using minimum total energy $\sum_{g\in G}x_g$.
We call this problem Art Gallery Problem with Fading, $\AGCrho(G,W)$.
When $\rho$ is clear from context, we occasionally write $\AGC(G,W)$ instead.
  In general, we are interested in
illuminating all of~$P$. Therefore, the relevant cases are $\AGC(G,P)$\dash---where
possible guard positions are given, e.g.\ $\AGC(V,P)$ is the common
discretization with vertex guards\dash---and $\AGC(P,P)$, the \emph{continuous/point guard} variant where guards can
be placed anywhere in~$P$.
$\AGCrho(G,W)$ can be formulated as a Linear Program~(LP):
\begin{alignat}{3}
  \text{\makebox[3em][l]{min}} & \sum_{g\in G } x_g \label{eq.dc.lpf.obj}\\
  \text{\makebox[3em][l]{\st}} & \sum_{g\in G\cap\vis{w}} \rho(g,w) x_g \geq 1 &\;\;& \forall w\in W\label{eq.dc.lpf.cover}\\
           &  x_g \geq 0                &    & \forall g\in G \;.\label{eq.dc.lpf.var}
\end{alignat}

It should be noted that $\AGC(P,P)$ results in an LP with an infinite number
of both variables and constraints.
Baumgartner et~al.~\cite{kbfs-esbgagp-12} were the first to show how to solve this
type of LP for the original AGP with fractional guards, i.e., for $\rho(\cdot, \cdot) = 1$.
 
\section{Fading Function and Complexity}\label{sec:comp}

\subsection{Fading Function}

\begin{figure}
	\begin{center}
		\def\svgwidth{\linewidth} 		\begingroup  \makeatletter  \providecommand\color[2][]{    \errmessage{(Inkscape) Color is used for the text in Inkscape, but the package 'color.sty' is not loaded}    \renewcommand\color[2][]{}  }  \providecommand\transparent[1]{    \errmessage{(Inkscape) Transparency is used (non-zero) for the text in Inkscape, but the package 'transparent.sty' is not loaded}    \renewcommand\transparent[1]{}  }  \providecommand\rotatebox[2]{#2}  \ifx\svgwidth\undefined    \setlength{\unitlength}{566.20068359bp}    \ifx\svgscale\undefined      \relax    \else      \setlength{\unitlength}{\unitlength * \real{\svgscale}}    \fi  \else    \setlength{\unitlength}{\svgwidth}  \fi  \global\let\svgwidth\undefined  \global\let\svgscale\undefined  \makeatother  \begin{picture}(1,0.21395028)    \put(0,0){\includegraphics[width=\unitlength,page=1]{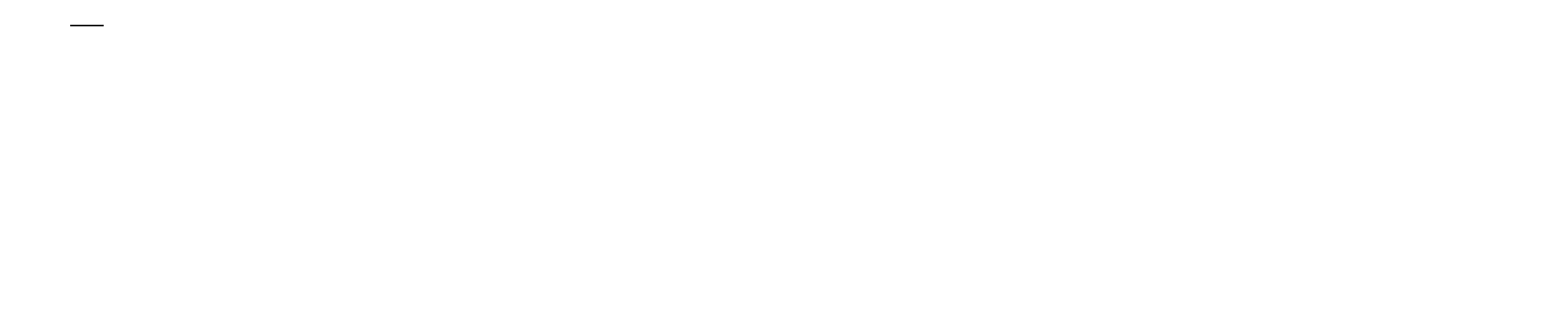}}    \put(0.04137479,0.16888073){\color[rgb]{0,0,0}\makebox(0,0)[rb]{\smash{$1$}}}    \put(0,0){\includegraphics[width=\unitlength,page=2]{rho.pdf}}    \put(0.16613574,0.00311949){\color[rgb]{0,0,0}\makebox(0,0)[b]{\smash{$1$}}}    \put(0,0){\includegraphics[width=\unitlength,page=3]{rho.pdf}}    \put(0.2587275,0.00311949){\color[rgb]{0,0,0}\makebox(0,0)[b]{\smash{$2$}}}    \put(0,0){\includegraphics[width=\unitlength,page=4]{rho.pdf}}    \put(0.44230653,0.00311949){\color[rgb]{0,0,0}\makebox(0,0)[b]{\smash{$4$}}}    \put(0,0){\includegraphics[width=\unitlength,page=5]{rho.pdf}}    \put(0.81354355,0.00311949){\color[rgb]{0,0,0}\makebox(0,0)[b]{\smash{$8$}}}    \put(0,0){\includegraphics[width=\unitlength,page=6]{rho.pdf}}    \put(1.00131602,0.00035013){\color[rgb]{0,0,0}\makebox(0,0)[rb]{\smash{distance}}}    \put(0.30570918,0.1124937){\color[rgb]{0,0,0}\makebox(0,0)[lb]{\smash{$\rho$}}}  \end{picture}\endgroup 	\end{center}
	\caption{Fading function~$\rho$.}
	\label{fig:rho}
\end{figure}

For a fading exponent $\alpha \in \R_{\geq 0}$, we define the fading function $\rho(g,w)$ by
\begin{equation}
	\rho(g,w) := \begin{cases}
		\|g-w\|^{-\alpha} & \text{if $\|g-w\| \geq 1$ and} \\
		1                 & \text{otherwise,}
	\end{cases}
	\label{eq:rho}
\end{equation}
compare Figure~\ref{fig:rho}.
Here $\|\cdot\|$ denotes the Euclidean norm in~$\R^2$.
As motivated above, the amount of light received by a point $w$ from the light source $g$ is $\rho(g,w) x_g$ with $\alpha = 2$ for natural light, $\alpha = 1$ for laser scanners, and $\alpha = 0$ for the classical~AGP.

Our motivation for capping $\rho$ at $1$ is two-fold:
(1)~Without this restriction, any non-zero guard provides an infinite amount of light to itself.
This results in the ``optimal solution'' of $\AGC(P,P)$ to consist of every point in $P$ glowing with an infinitesimally small brightness.
This problem is not well-defined, see Lemma~\ref{lem:altrho-no-opt}.
(2)~In an application, guards correspond physical objects like light bulbs or scanners.
These objects are not point-shaped; they have a physical size.
We are interested in obtaining coverage for the environment of the light source rather than for the light source itself.
Hence, we assume that the input is scaled such that the radius $1$ around a guard falls completely into the object size.
That motivates capping $\rho$ at~1, even in the discrete version.

\begin{lemma}\label{lem:altrho-no-opt}
  Let $\tilde{\rho}(g,w) := \|g-w\|^{-\alpha}$, $\alpha>0$, be an
  alternate fading function to $\rho(g,w)$ in Equation~\eqref{eq:rho}.
  Then $\operatorname{AGPF\tilde{\rho}}(P,P)$ has no optimal solution.
\end{lemma}
\begin{proof}
  We show that for any finite $G \subset P$ any illumination $x \in \R^G_{\geq 0}$
  either is infeasible or not optimal because it can be improved by using more guards $G' \supset G$.
  The dual of $\operatorname{AGPF\tilde{\rho}}(P,P)$ is
  \begin{alignat}{3}
    \text{\makebox[3em][l]{max}} & \sum_{w\in P} y_w\\
    \text{\makebox[3em][l]{\st}} & \sum_{w\in \vis{g}} \altrho(g,w) y_w \leq 1 &\;\;& \forall g\in P\label{eq:altrho.cons}\\
    & y_w \geq 0              &    & \forall w\in P\;,
  \end{alignat}
  the problem of fractionally packing bright witnesses in $P$ such that no
  guard receives too much light.

  Assume we are given a feasible, optimal solution $x^*$ for $\operatorname{AGPF\tilde{\rho}}(P,P)$ that only
  uses a finite number of guards.
  By feasibility, there exists a guard $g\in P$ with $x^*_g>0$.
  Moreover, by optimality, there exists a witness $w_p\in W$ with
  \begin{equation*}
	\sum_{g\in G\cap\vis{w_p}} \rho(g,w_p) x^*_g = 1 \;,
\end{equation*}
  otherwise we could reduce the value of some $x^*_g$ while preserving feasibility.
  Since Constraint~\eqref{eq.dc.lpf.cover} holds with equality for~$w_p$, we have $y_{w_p} >0$ in the dual.
  Now consider a point $g\in P$ converging towards~$w_p$. As
  $y_{w_p}>0$ and $\altrho(g,w_p)\to\infty$ as $g\to w_p$, at some
  point
  \begin{equation*}
  \sum_{w\in W\cap\vis{g}} \altrho(g,w) y_w \geq \altrho(g,w_p) y_{w_p} > 1
  \end{equation*}
  must hold. A point $g$ fulfilling this inequality indicates a violated Constraint~\eqref{eq:altrho.cons}.
  Hence, $y$ is dually infeasible, so $x^*$ is not optimal, contradicting our assumptions.
\end{proof}

This shows that whenever $G$ is finite, we can find an additional
guard that improves the solution. Consequently, $\operatorname{AGPF\tilde{\rho}}(P,P)$
only possesses ``optimal solutions'' with infinitely many guards. In fact, it
is easy to see that the ``optimal'' guard positions $G$ are topologically dense in~$P$,
using the above $\varepsilon$-neighborhood argument. Then,
the solution values can be scaled down arbitrarily since
$\altrho(g,w)$ approaches $\infty$ as $w\to g$. This way, we
can construct a series of solutions of strictly decreasing values,
where $G$ is infinite and the values of all guards converge to~0.
Obviously, the limit of this process is $x\equiv 0$ which is
not feasible.

Note that while we cap $\rho$ at~1, Eisenbrand et~al.\ do not employ capping in~\cite{efkm-esi-05}:
Given a finite set of guard candidates $S$, they illuminate a stage $L$.
The guard candidates usually are not located on the stage. Therefore, no point infinitesimally close to a guard needs to be covered.
In addition, even if guard candidates were located on~$L$, the set $S$ is not dense in~$L$.
Hence, they never experience the situation described by Lemma~\ref{lem:altrho-no-opt}.

\subsection{Complexity Results}

We draw two conclusions regarding the complexity of the~AGPF:
\edef\lemmahardness{\thelemma}
\begin{lemma}\label{lem:hardness}
  For a finite set $G\subset P$, finding integer solutions to
  $\AGC(G,P)$ is NP-hard. Finding integer solutions for $\AGC(P,P)$ is APX-hard.
\end{lemma}

\begin{proof}
  For $\alpha=0$, finding integer solutions corresponds to solving
  the original AGP\@. The corresponding hardness results~\cite{esw-irgpt-01,r-agta-87}
  apply.
\end{proof}

It is tempting to solve the $\AGC$ by reducing it to the $\AG$.
Unfortunately, this approach results in an approximation factor which is
exponential in $\alpha$:

\begin{figure*}
\def\svgwidth{\textwidth}
\begingroup
  \makeatletter
  \providecommand\color[2][]{    \errmessage{(Inkscape) Color is used for the text in Inkscape, but the package 'color.sty' is not loaded}
    \renewcommand\color[2][]{}  }
  \providecommand\transparent[1]{    \errmessage{(Inkscape) Transparency is used (non-zero) for the text in Inkscape, but the package 'transparent.sty' is not loaded}
    \renewcommand\transparent[1]{}  }
  \providecommand\rotatebox[2]{#2}
  \ifx\svgwidth\undefined
    \setlength{\unitlength}{285.51831055pt}
  \else
    \setlength{\unitlength}{\svgwidth}
  \fi
  \global\let\svgwidth\undefined
  \makeatother
  \begin{picture}(1,0.05268242)    \put(0,0){\includegraphics[width=\unitlength]{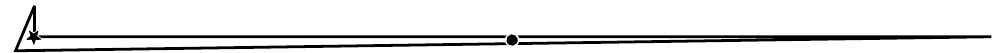}}    \put(0.0162205,0.01579766){\color[rgb]{0,0,0}\makebox(0,0)[rb]{\smash{$1$}}}    \put(0.36556384,0.0205041){\color[rgb]{0,0,0}\makebox(0,0)[b]{\smash{$s$}}}  \end{picture}\endgroup
 \caption{\label{godfried} An optimal guard placement for the AGP~($\star$) is not optimal for the AGP with fading ($\star$~and~$\bullet$).}
\label{fig:godfried}
\end{figure*}

\edef\lemmanoapx{\thelemma}
\begin{lemma}\label{lem:no-apx}
  The AGPF cannot be approximated by a better factor than $2^\alpha$ by first solving
  AGP and computing intensities afterwards.
\end{lemma}
\begin{proof}
	Consider the polygon shown in Figure~\ref{fig:godfried} with a spike of length~$s$.
	The optimal covering guard set consists of one guard~($\star$).
	Illuminating the polygon with it requires an intensity of~$s^\alpha$.
	An optimal illumination uses two lights ($\star$~and~$\bullet$) with intensities 1 and about $\left(\tfrac{1}{2}s\right)^\alpha$.
	Therefore, the approximation factor of the first solution is at least
	\begin{equation*}
		\frac{s^\alpha}{1+\left(\tfrac{1}{2}s\right)^\alpha}
	\end{equation*}
	which converges to $2^\alpha$ for large~$s$.
\end{proof}

Note that this result holds for $\AGC(V,P)$, as we can simply integrate a small bend in the lower edge with a reflex vertex in the height of $\bullet$. Again, the optimal solution to
the AGP restricted to vertex guards would be $\star$. Two lights at $\star$ and the new reflex vertex at $\bullet$ result in the same bound on the approximation factor from Lemma~\ref{lem:no-apx}.

\section{Algorithms for Fixed (Vertex) Guards}
\label{sec:fixedguards}

In this section, we turn our attention to $\AGC(G,P)$, the AGPF with a fixed, finite set of guard candidates instead of its point guard sibling $\AGC(P,P)$.
The reasoning behind that decision is two-fold:

(1)~Although both, the AGP with point guards and the AGP with vertex guards, are NP-hard~\cite{ll-ccagp-86},
many papers identified the vertex guard variant of the AGP to be much easier solvable in terms of practical algorithms than the point guard variant~\cite{TR-IC-09-46,crs-exmvg-11,rsfhkt-eag-14,kbfs-esbgagp-12}.
Since the AGPF is closely related to the AGP and our algorithms stem from our experience with the previous algorithms, see~\cite{rsfhkt-eag-14} for an overview, we start our investigation by tackling the discrete version.

(2)~The hardness in terms of tractability by practical algorithms of the AGPF lies in finding the darkest point in intermediate solutions (in this paper referred to as Primal Separation Problem~(PSP)).
Generating good/optimal guard locations has emerged to be the hardest part of the AGP with point guards~\cite{rsfhkt-eag-14}.
So we refrain from entwining the PSP and guard generation to obtain meaningful insights on the tractability of the (novel) PSP.

Staying in line with previous work and references therein, see~\cite{rsfhkt-eag-14} for an overview, it is safe to assume $G = V$.
Our algorithms, however, only require $\bigcup_{g \in G}\V(g) = P$, i.e., that the instance is feasible at all.

This section proposes two algorithms for $\AGC(G,P)$.
One of them, see Section~\ref{sec:appx}, replaces the continuous fading function $\rho$ by an appropriately chosen step function and yields a FPTAS for $\AGC(G,P)$.
The other one, in Section~\ref{sec:lipschitz}, works directly on the continuous fading function $\rho$ and uses simplex partitioning to identify the darkest point for a given illumination scheme.

\subsection{Discrete Approximation Algorithm and FPTAS}
\label{sec:appx}

\begin{figure}
	\begin{center}
		\def\svgwidth{\linewidth} 		\begingroup  \makeatletter  \providecommand\color[2][]{    \errmessage{(Inkscape) Color is used for the text in Inkscape, but the package 'color.sty' is not loaded}    \renewcommand\color[2][]{}  }  \providecommand\transparent[1]{    \errmessage{(Inkscape) Transparency is used (non-zero) for the text in Inkscape, but the package 'transparent.sty' is not loaded}    \renewcommand\transparent[1]{}  }  \providecommand\rotatebox[2]{#2}  \ifx\svgwidth\undefined    \setlength{\unitlength}{566.20068359bp}    \ifx\svgscale\undefined      \relax    \else      \setlength{\unitlength}{\unitlength * \real{\svgscale}}    \fi  \else    \setlength{\unitlength}{\svgwidth}  \fi  \global\let\svgwidth\undefined  \global\let\svgscale\undefined  \makeatother  \begin{picture}(1,0.21395028)    \put(0,0){\includegraphics[width=\unitlength,page=1]{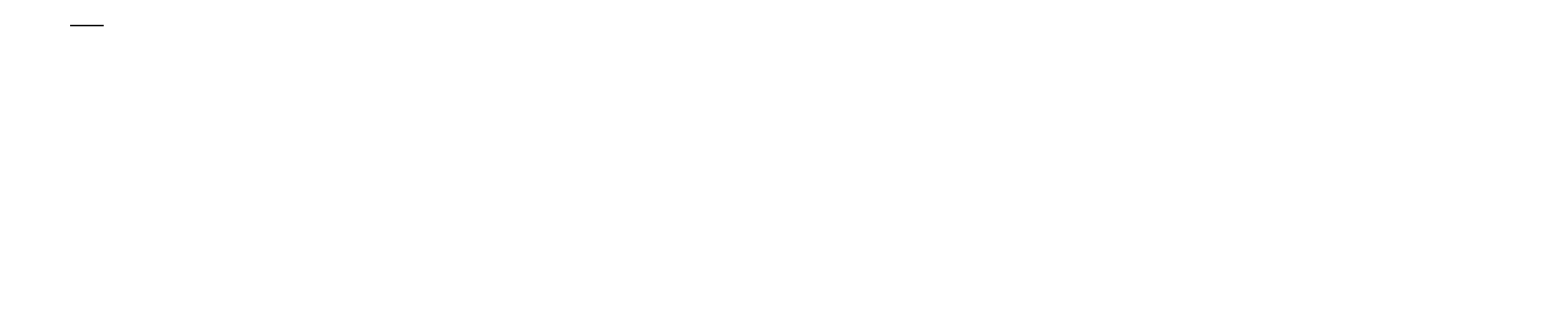}}    \put(0.04137479,0.16888073){\color[rgb]{0,0,0}\makebox(0,0)[rb]{\smash{$1$}}}    \put(0,0){\includegraphics[width=\unitlength,page=2]{rho_tau.pdf}}    \put(0.16613574,0.00311949){\color[rgb]{0,0,0}\makebox(0,0)[b]{\smash{$1$}}}    \put(0,0){\includegraphics[width=\unitlength,page=3]{rho_tau.pdf}}    \put(0.2587275,0.00311949){\color[rgb]{0,0,0}\makebox(0,0)[b]{\smash{$2$}}}    \put(0,0){\includegraphics[width=\unitlength,page=4]{rho_tau.pdf}}    \put(0.44230653,0.00311949){\color[rgb]{0,0,0}\makebox(0,0)[b]{\smash{$4$}}}    \put(0,0){\includegraphics[width=\unitlength,page=5]{rho_tau.pdf}}    \put(0.81354355,0.00311949){\color[rgb]{0,0,0}\makebox(0,0)[b]{\smash{$8$}}}    \put(0,0){\includegraphics[width=\unitlength,page=6]{rho_tau.pdf}}    \put(1.00131602,0.00035013){\color[rgb]{0,0,0}\makebox(0,0)[rb]{\smash{distance}}}    \put(0.30570918,0.1124937){\color[rgb]{0,0,0}\makebox(0,0)[lb]{\smash{$\rho$}}}    \put(0.22540378,0.06703009){\color[rgb]{1,0,0}\makebox(0,0)[rb]{\smash{$\tau$}}}  \end{picture}\endgroup 	\end{center}
	\caption{Fading function $\rho$  and step function $\tau$ for $\epsilon=1$.}
	\label{fig:rho-tau}
\end{figure}

In this section, we introduce \Tdiscrete, a discretized approximation algorithm
for $\AGC(G,P)$ that yields a FPTAS. We achieve this by replacing
the continuous fading function $\rho$ with a step function $\tau$,
compare Figure~\ref{fig:rho-tau}.
For a fixed $\epsilon > 0$, $\tau$ is defined as follows:
\begin{equation}
  \tau(g,w) := \lhfloor \rho(g,w) \rhfloor_{1+\epsilon},
  \label{eq:taudef}
\end{equation}
where $\lhfloor\cdot\rhfloor_b$ denotes the hyperfloor with base
$b$, i.e., $\lhfloor x\rhfloor_b := \max\{ b^z:\;b^z\leq x,\;
z\in\Z \}$. Therefore, $\tau$ is piecewise constant and
approximates $\rho$, i.e.,
\begin{equation}
  1/(1+\varepsilon)\;\rho < \tau \leq \rho\;.
  \label{eq:tauapx}
\end{equation}
We can use $\AGCstep(G,W)$ to obtain a $(1+\epsilon)$-approximation for $\AGC(G,W)$:

\edef\lemmatauapx{\thelemma}
\begin{lemma}\label{lem:tau-apx}
  A feasible solution for $\AGCstep(G,W)$ is also feasible for
  $\AGCrho(G,W)$. An optimal solution for $\AGCstep(G,W)$ is a
  $(1+\epsilon)$-approximation for $\AGCrho(G,W)$.
\end{lemma}
\begin{proof}
  The proof follows from Inequality~\eqref{eq:tauapx}:
  \begin{enumerate}
  \item Since $\rho \geq \tau$, a solution fulfilling inequalities of type~\eqref{eq.dc.lpf.cover} with~$\tau$, does so too using~$\rho$. Hence every $\AGCstep(G,W)$ solution also solves $\AGCrho(G,W)$.
  \item Consider an optimal solution $x^*$ of $\AGCrho(G,W)$ with objective value $\OPT$. It follows from~\eqref{eq:tauapx} that $x_\tau:=(1+\epsilon)x^*$, i.e., scaling every guard with a factor $1+\epsilon$, is a feasible solution to $\AGCstep(G,W)$.
  $x_\tau$ has objective value $(1+\epsilon)\OPT$, yielding the claimed approximation factor.
  \end{enumerate}
Note that our arguments also hold for $G = W = P$.
\end{proof}

\subsubsection{An Algorithm for $\AGCstep$}\label{sec:appx.alg}

To solve $\AGC_\tau(G,P)$, we consider the illumination function
\begin{equation}
  T_{G,x}(p) := \sum_{g\in G \cap \vis{p}} \tau(g,p)x_g\,,
  \label{eq:def-illum-function}
\end{equation}
which defines the amount of light received at any point $p\in P$
w.r.t.\ to the fading function $\tau$, given a solution $x \in
\R^{G}_{\geq 0}$. $x$ is feasible if and only if $T_{G,x} \geq 1$.
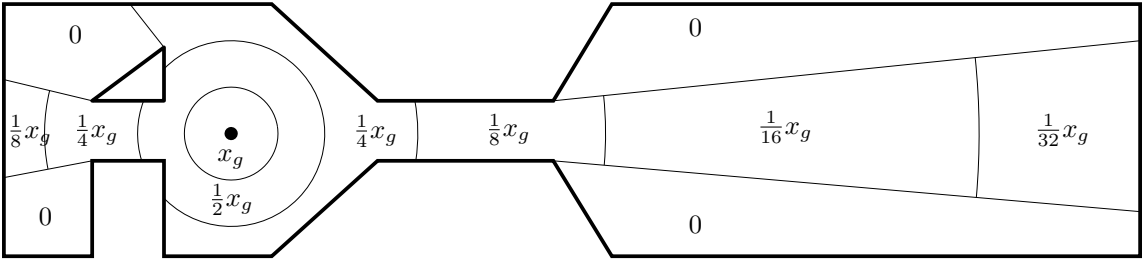
\begin{figure*}
    \def\svgwidth{\textwidth}
\begingroup
  \makeatletter
  \providecommand\color[2][]{    \errmessage{(Inkscape) Color is used for the text in Inkscape, but the package 'color.sty' is not loaded}
    \renewcommand\color[2][]{}  }
  \providecommand\transparent[1]{    \errmessage{(Inkscape) Transparency is used (non-zero) for the text in Inkscape, but the package 'transparent.sty' is not loaded}
    \renewcommand\transparent[1]{}  }
  \providecommand\rotatebox[2]{#2}
  \ifx\svgwidth\undefined
    \setlength{\unitlength}{487.475pt}
  \else
    \setlength{\unitlength}{\svgwidth}
  \fi
  \global\let\svgwidth\undefined
  \makeatother
  \begin{picture}(1,0.22611416)    \put(0,0){\includegraphics[width=\unitlength]{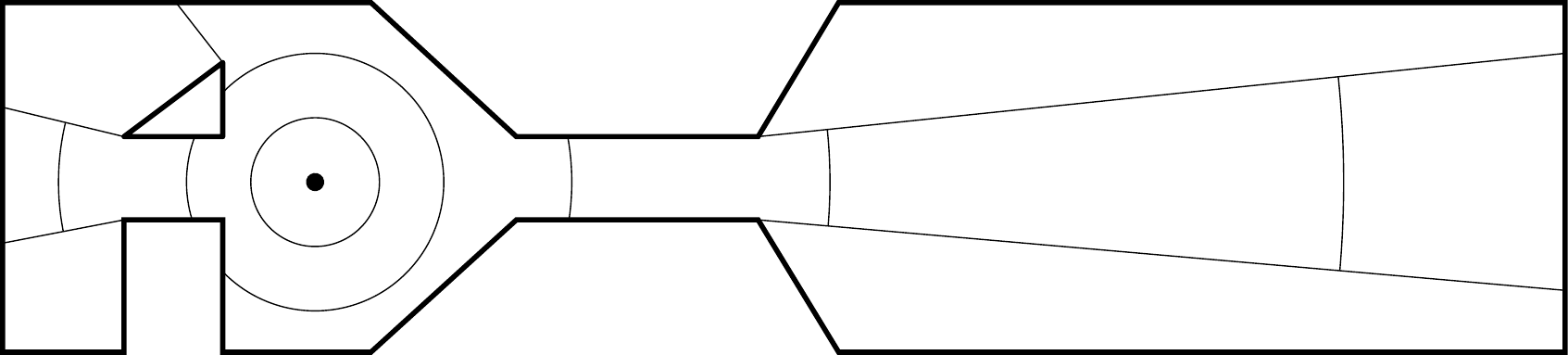}}    \put(0.20105373,0.08597366){\color[rgb]{0,0,0}\makebox(0,0)[b]{\smash{$x_g$}}}    \put(0.20105373,0.04574328){\color[rgb]{0,0,0}\makebox(0,0)[b]{\smash{$\tfrac{1}{2}x_g$}}}    \put(0.32766905,0.10644286){\color[rgb]{0,0,0}\makebox(0,0)[b]{\smash{$\tfrac{1}{4}x_g$}}}    \put(0.08262061,0.10702898){\color[rgb]{0,0,0}\makebox(0,0)[b]{\smash{$\tfrac{1}{4}x_g$}}}    \put(0.004656,0.10641847){\color[rgb]{0,0,0}\makebox(0,0)[lb]{\smash{$\tfrac{1}{8}x_g$}}}    \put(0.44404416,0.10761509){\color[rgb]{0,0,0}\makebox(0,0)[b]{\smash{$\tfrac{1}{8}x_g$}}}    \put(0.68679699,0.10995952){\color[rgb]{0,0,0}\makebox(0,0)[b]{\smash{$\tfrac{1}{16}x_g$}}}    \put(0.93011451,0.10714677){\color[rgb]{0,0,0}\makebox(0,0)[b]{\smash{$\tfrac{1}{32}x_g$}}}    \put(0.60828029,0.02202652){\color[rgb]{0,0,0}\makebox(0,0)[b]{\smash{$0$}}}    \put(0.60828029,0.19650727){\color[rgb]{0,0,0}\makebox(0,0)[b]{\smash{$0$}}}    \put(0.06452717,0.18977526){\color[rgb]{0,0,0}\makebox(0,0)[b]{\smash{$0$}}}    \put(0.0381522,0.02918093){\color[rgb]{0,0,0}\makebox(0,0)[b]{\smash{$0$}}}  \end{picture}\endgroup
 \caption{The visibility Arrangement $\calA(\{g\})$ induced by
$T_{\{g\},x}$ for a point $g$ with value
  $x_g$, assuming $\alpha=1$.}
\label{fig:arr}
\end{figure*}
It is easy to see that, for a single guard $g$, $T_{\{g\},x}$ is a
piecewise constant function over $P$, as shown in
Figure~\ref{fig:arr}. It induces an arrangement $\calA(\{g\})$
within $P$ with faces of constant value. $P$ is split into two
parts:
One is $P\setminus\vis{g}$, where $T_{\{g\},x}$ is constant $0$.
For the remaining part, note that $\tau$ is a monotonically decreasing step function taking the values
$(1+\epsilon)^{-z}$ for $z=0,1,2,\ldots$, each at distances
$((1+\varepsilon)^{z-1},(1+\varepsilon)^z]$ from $g$ (exception:
$[0,1]$ for $z=0$).  This introduces $\bigO{\log_{1+\epsilon} D}$
concentric bands of equal coverage value around~$g$. Therefore,
this part of the arrangement is the intersection of
concentric circles with $\vis{g}$.

Now consider the arrangement $\calA(G)$ defined by overlaying the
individual arrangements for all $g\in G$. Then, $T_{G,x}$ is
constant over each face, edge, and vertex of this arrangement.
This arrangement has a bounded complexity:
\edef\lemmacalacomplexitybound{\thelemma}
\begin{lemma}\label{lem:cala-complexity-bound}
  The complexity of $\calA(G)$ is
  $\bigO{(n+\log_{1+\epsilon}D)^2|G|^2}$.
\end{lemma}
\begin{proof}
  The arrangement $\calA(G)$ is constructed from
  \begin{itemize}
  \item $n$ straight line segments defining $P$,
  \item $\bigO{n}$ straight line segments defining $\vis{g}$ per $g\in
    G$, and
  \item $\bigO{\log_{1+\epsilon}D}$ circles per $g\in G$,
  \end{itemize}
  therefore, $\bigO{(n+\log_{1+\epsilon}D)|G|}$ straight line segments and
  circles in total.
Squaring this figure accounts for intersections in $\mathcal{A}(G)$.
\end{proof}

We are ready to present the overall algorithm (referred to as \Tdiscrete):

\begin{description}
\item [\bf{Step 0 (Visibility Arrangement).}]
    First compute the overlay of all guards' visibility step-functions, $\calA(G)$.
    It has polynomial complexity by Lemma~\ref{lem:cala-complexity-bound}.

\item [\bf{Step 1 (Witness Points).}]
    Place one witness $w$ in each feature (vertex, edge, and face interior of $\calA(G)$).
    Let $W$ denote the set of these witnesses, which is polynomial in cardinality.
\item [\bf{Step 2 (Solve LP).}]
    Solve $\AGCstep(G,W)$, which is possible in polynomial time.
    The solution $x$ is feasible and optimal for $\AGCstep(G,P)$.
\end{description}

$\AGCstep(G,P)$, our original LP, has a finite number of
variables, and an infinite number of constraints since every
point in $P$ has to be illuminated. But we can exploit that
$T_{G,x}$ is constant in each feature of the overlay in our LP
formulation: An arbitrary witness point in an arrangement feature
$f$ is sufficiently lit if and only if all of $f$ is sufficiently
lit as well. Hence, we can solve the finite LP $\AGCstep(G,W)$,
where $W$ contains one witness point in every feature of
$\calA(G)$.

Putting everything together, we get

\edef\theoremfptas{\thetheorem}
\begin{theorem}\label{thm:fptas}
  $\AGC(G,P)$ admits a fully polynomial-time approximation scheme.
\end{theorem}
\begin{proof}
  Given an $\epsilon>0$, the arrangement $\calA(G)$ has polynomial complexity due to
  Lemma~\ref{lem:cala-complexity-bound} and the fact that
  $\log_{1+\epsilon}D$ is polynomially bounded in the encoding size of
  $\epsilon$ and $P$. Hence the LP $\AGCstep(G,W)$ in the discrete algorithm can
  be solved in polynomial time. According to Lemma~\ref{lem:tau-apx},
  the result is a $(1+\epsilon)$-approximation to $\AGC(G,P)$.
\end{proof}

\subsubsection{Implementation Issues and Approximation of Circular Arcs}
\label{sec:circles-octagons}

Although $\calA(G)$ and $W$ are polynomial in size, they become
very large from a practical point of view. Therefore, it is more
efficient to solve this with the technique of iterative
primal separation
(compare~\cite{csr-eeeaoagp-08,rsfhkt-eag-14,kbfs-esbgagp-12} for
the classical AGP), i.e., to leave out some constraints of the LP
and start with a small set $W$, then solve $\AGCstep(G,W)$, and
then check the $\calA(G)$ for insufficiently lit points. If no
such points exist, the solution is feasible. Otherwise, we add
representative witness points for insufficiently lit features to
$W$ and re-iterate the process until the solution is feasible. We
know this approach terminates because in each iteration at least
one new feature of $\calA(G)$ is covered.

It should be noted that by introducing circular arcs in the
arrangement $\calA(G)$, we have to deal with irrational
numbers\dash---to be precise, with algebraic numbers of degree~2.
This is known to result in computationally expensive operations.
Hence, we can replace the circular arcs by regular octagons.
They still have irrational coordinates:
For example, consider a regular octagon centered at the origin with its vertices on a circle of
radius $r$. It has a vertex at
$r (\cos \frac{\pi}{4}, \sin \frac{\pi}{4}) = r(1 / \sqrt 2, 1 / \sqrt 2)$.
However, intersection tests in the arrangement only
use linear functions, which are expected to be less expensive.
The octagon-based approach somewhat reduces the approximation guarantees\dash---octagons only approximate circles\dash---but still yields a FPTAS.
However, our experimental evaluation shows that using the exact support for circular arcs of CGAL~\cite{cgal,cgal:cpt-cgk2-15a}
outperforms the octagon-based approach.

\subsection{Continuous Optimization Approach}
\label{sec:lipschitz}

In this section, we introduce \Tcontinuous, an alternative algorithm for $\AGC(G,P)$.
\Tcontinuous directly uses the continuous fading function $\rho$ instead of an approximate step function.
It iteratively solves $\AGCrho(G,W)$, starting with $W = \emptyset$ and checks if its solution is feasible for $\AGC(G,P)$.
If a point $w \in P \setminus W$ that currently does not receive sufficient light can be found, it is added to~$W$;
the process is repeated until $P$ is sufficiently lit.

This requires a subroutine for finding insufficiently lit points and can obviously be answered by a routine that finds the darkest point.
Hence, we seek an algorithm solving the \emph{Primal Separation Problem~(PSP),} i.e., given a current solution~$x$, for finding:
\begin{equation}
	\argmin{w \in P} T_{G,x}(w)
		= \argmin{w \in P} \sum_{g \in G \cap \V(w)} \rho(g,w) x_g\,.
	\label{eq:PSP}
\end{equation}
For the sake of presentation, we simplify $T(w) := T_{G,x}(w)$.

We solve the PSP in two steps.
First, we triangulate the overlay of all visibility polygons $\vis g$ for all $g\in G$
which yields a collection of triangles, each of which seen by a fixed subset of guards.
In Section~\ref{sec:simplex}, we present an algorithm to solve~\eqref{eq:PSP} in such a
triangle $\calS \subseteq P$. The algorithm requires a function $\ell$ that provides a lower
bound on the minimum brightness of $T_{G,x}$ in~$\calS$, i.e., $\ell(\calS) \leq
T(p)\;\forall p\in\calS$.
We discuss two such functions in Sections~\ref{sec:geom-lb} and~\ref{sec:lip-lb}.

\subsubsection{Simplex Partitioning Algorithm}
\label{sec:simplex}

\Tcontinuous employs a global optimization approach based on simplex (here: triangle) partitioning~\cite{go-action-1996} to solve the PSP.
It iteratively keeps track of
\begin{itemize}
\item $\mathcal{S}_0$, the triangle to be searched for a darkest spot (the input),
\item $k = 1, 2, 3, \ldots$, the iteration counter,
\item $\mathcal{U}_{k}$, a set of triangles with $\bigcup_{\mathcal{S} \in \mathcal{U}_{k}} \mathcal{S} \subseteq \mathcal{S}_0$ (the remaining search region),
\item $\mathcal{S}_k$, the next triangle to be partitioned,
\item $x_{k}^{\ast}$, the darkest point found so far (the incumbent solution), and
\item $\beta_k$, the current lower bound on the minimum brightness in~$\mathcal{S}_0$.
\end{itemize}

\begin{description}
\item[\bf{Step 0 (Initialization).}]
    Set $\mathcal{U}_0 = \{\mathcal{S}_0\}$ as the current set of triangles.
            Set the current best solution  to
    $x_{0}^{\ast} = \arg\min_{v\in V(\mathcal{S}_0)} T(v)$ (the darkest vertex of $\mathcal{S}_0$).
    Compute a lower bound $\beta_{0} = \ell(\mathcal{S}_{0})$ for
    the minimum brightness in~$\mathcal{S}_0$.
\end{description}
In each iteration $k \in \N$, we do the following steps:
\begin{description}
\item[\bf{Step 1 (Eliminating bright triangles).}]
                    Let $\mathcal{R}_{k}$ be the collection of
    triangle candidates remaining after discarding all triangles with $\ell(\mathcal{S}) > T(x_{k-1}^{\ast})$ from $\mathcal{U}_{k-1}$, i.e., after removing the triangles that cannot contain a darkest point due to their lower bound.

\item[\bf{Step 2 (Selecting a triangle).}]
    Consider the current collection of triangles $\mathcal{R}_{k}$ and select
    \begin{equation}
    \label{eq:DSP-lower-bound-8}     \mathcal{S}_k = \argmin{\mathcal{S} \in \mathcal{R}_{k}} \ell(\mathcal{S})\,,
    \end{equation}
    the triangle with the smallest lower bound, to be partitioned.

\item[\bf{Step 3 (Refining the search region).}]
    Divide $\mathcal{S}_{k}$ into two smaller triangles
    by applying a bisection on its longest edge such that
    \begin{equation}
    \label{eq:DSP-lower-bound-9}     \mathcal{S}_{k} = \mathcal{S}_{k_1} \cup \mathcal{S}_{k_2}.
    \end{equation}

\item[\bf{Step 4 (Updating the incumbent).}]
    At this point, we obtain the search region
    \begin{equation}
    \label{eq:DSP-lower-bound-10}     \mathcal{U}_{k} = (\mathcal{R}_{k} \setminus \{\mathcal{S}_{k}\}) \cup \{\mathcal{S}_{k_1}, \mathcal{S}_{k_2}\}.
    \end{equation}
    Evaluate $T$ at the common new vertex $v_k$ of the triangles $\mathcal{S}_{k_1}$ and
    $\mathcal{S}_{k_2}$.
    Update
    \begin{align}
    x_{k}^{\ast} &= \argmin{v \in \{v_k, x_{k-1}^{\ast}\}} T(v)\,, \\
    \beta_k &= \min_{\mathcal{S} \in \mathcal{U}_{k}} \ell(\mathcal{S})\,.
    \end{align}
    Then increase $k$ and continue at Step~1.

  \item[\bf{Stopping criterion.}]
	Stop when there is no triangle to be partitioned.
		Alternatively, stop after a certain number of iterations or when a $\delta$-estimation of the
optimal solution is found, i.e., when $T(x_{k}^{\ast}) \leq \beta_k + \delta$.

\end{description}

\medskip
\noindent
We propose two approaches for computing the lower bounds, i.e.,
$\ell(\cdot)$, in Sections~\ref{sec:geom-lb} and~\ref{sec:lip-lb}. The
first exploits geometric properties and the second the fact that
$T$ is Lipschitzian.

\subsubsection{Geometric Lower Bound}
\label{sec:geom-lb}

Regarding a lower bound used to prune triangles in the simplex-partitioning algorithm in Section~\ref{sec:simplex}, observe that $\rho(g,w)$ monotonically decreases with increasing $\|g-w\|$, i.e., light decreases with increasing distance.
Fix a triangle $\mathcal S$, a finite set of guards $G$ scattered in~$P$, and $x \in \R^G_{\geq 0}$.
According to the definition of $\rho(g,w)$, $T_{\{g\},x}$ attains its minimum at a vertex of~$\mathcal S$.
Therefore,
\begin{equation}
  \label{eq:geomlowerbound}
  	\ell(\mathcal{S}) := \sum_{g \in G} \;\min_{p \text{ vertex of } \mathcal{S}} T_{\{g\},x}(p)
\end{equation}
is a lower bound for $\min_{p \in \mathcal{S}} T_{G,x}(p)$.

\subsubsection{Lipschitz Lower Bound}
\label{sec:lip-lb}

As an alternative approach to the geometric estimation of the
lower bound, one may use the analytical properties of $T$. For
example, if $T$ is a Lipschitzian with a known Lipschitz constant
$L$, then one can use the following result (see
\cite{go-action-1996}):
\edef\lemmalipschitzbound{\thelemma}
\begin{lemma}\label{lem:lipschitz-bound}
Let $V(S)= \{v_0,v_1,\dots,v_m\}$ be the vertex set of the simplex
$\mathcal{S}$. Let $f$ be a Lipschitzian on $\mathcal{S}$ with the
Lipschitz constant $L$. Denote by $z_j$ the function values
$T(v_m)$ (for $j=0,\dots,m$), then we have:
\begin{equation}
\label{eq:DSP-lower-bound-0} z^{\ast} \geq \frac{1}{m+1}\left(\sum\limits_{j=0}^{m}z_j - L
\max_{0 \leq i \leq m} \sum\limits_{j=0}^{m} \|v_i - v_j\|\right)
\end{equation}
\end{lemma}
\begin{proof}
The proof can be found in \cite{go-action-1996}; however, for the
sake of completeness, we provide the complete proof. Let $x$ be an
arbitrary point of the simplex $\mathcal{S}$. Since $T$ is a
	Lipschitz function (see Lemma~\ref{lem:psp-is-lipschitz})
\begin{equation}
\label{eq:DSP-lower-bound-1} |T(x) - z_{j}|\leq L\|x - v_j\|, \quad j=0,1,\dots,m.
\end{equation}
We know that this inequality is true for any point in
$\mathcal{S}$, particularly for the optimal point $x^{\ast}$;
consequently, we have
\begin{equation*}
\label{eq:DSP-lower-bound-2} z^{\ast}=T(x^{\ast}) \geq z_{j} - L\|x^{\ast} - v_j\|, \quad
j=0,1,\dots,m.
\end{equation*}
By summing up these inequalities, we arrive at the following
inequality
\begin{equation}
\label{eq:DSP-lower-bound-3} z^{\ast} \geq  \frac{1}{m+1} \left(\sum\limits_{j=0}^{m} z_{j} - L
\sum\limits_{j=0}^{m} \|x^{\ast} - v_j\|\right).
\end{equation}
Due to the fact that $x^{\ast}$ is unknown, we need to provide an
estimation of the second sum in (\ref{eq:DSP-lower-bound-3}).
Define the function $Q$ as follows:
\begin{equation*}
\label{eq:DSP-lower-bound-4} Q(x) := \sum\limits_{j=0}^{m} \|x - v_j\| \colon \quad x \in
\mathcal{S}.
\end{equation*}
Since $\mathcal{S}$ is a convex set and $Q$ is a convex function,
the maximum of $Q$ is attained at one of the vertices of
$\mathcal{S}$. Consequently,
\begin{equation*}
\label{eq:DSP-lower-bound-5} \sum\limits_{j=0}^{m} \|x^{\ast} - v_j\| \leq \max_{0 \leq i \leq
m} \sum\limits_{j=0}^{m} \|v_i - v_j\|.
\end{equation*}
Using this inequality in (\ref{eq:DSP-lower-bound-3}) completes
the proof.
\end{proof}

Fortunately, this bound can be used in the PSP:
\edef\lemmapspislipschitz{\thelemma}
\begin{lemma}\label{lem:psp-is-lipschitz}
  $T$ is Lipschitzian.
\end{lemma}
\begin{proof}
We provide the proof for $\alpha=2$. The other cases can be proved
in a similar way.
For the sake of simplicity in the notation, let us define for any~$i$ (such that $i=1,\dots,m$):
\begin{itemize}
    \item $\rho_{i}(w):=\rho(g_{i},w)$.    \item $(x_{0},y_{0}):=(x^{i},y^{i})$ as the coordinates of the point $g_{i}$.     \item $(x,y)$ as the coordinates of the witness point $w$ and
    $(x_{1},y_{1})$ as the coordinates of the witness point $w_{x}$ and
    $(x_{2},y_{2})$ as the coordinates of the witness point $w_{y}$
    (where (through the primal-dual procedure)
    $w_{x}$ and $w_{y}$ are the potential points that will be chosen to be added into the set of the witness points).
\end{itemize}
Thus
\begin{equation}\label{eq:rho-lip}
  	\rho_{i}(w) := \begin{cases}
      	\frac{1}{\sqrt{(x-x_{0})^2 + (y-y_{0})^2}} & \text{if $\|w-g_{i}\| \geq 1$,}       \\ 1 & \text{otherwise.}
    \end{cases}
\end{equation}
We want to show that $ \exists L > 0 \quad \mbox{ s.t. } \forall
w_{x},w_{y}: |\rho_{i}(w_{x}) - \rho_{i}(w_{y})| \leq L \| w_{x} -
w_{y} \|.$ We need to consider three different cases:
\begin{itemize}
    \item[(i)] $\| w_{x} - g_{i} \| \geq 1$ and $\| w_{y} - g_{i} \| \geq 1$,    \item[(ii)] $\| w_{x} - g_{i} \| < 1$ and $\| w_{y} - g_{i} \| < 1$, and    \item[(iii)] $\| w_{x} - g_{i} \| \geq 1$ and $\| w_{y} - g_{i} \| < 1$.\end{itemize}
\textbf{Case (i): } In this case, the following definition of
$\rho_{i}(\cdot)$ is applied:
\[
\rho_{i}(w) = \frac{1}{\sqrt{(x-x_{0})^2 + (y-y_{0})^2}} : w \in
\{w_{x},w_{y}\}.
\]
It is sufficient to show that the partial derivatives of
$\rho_{i}(w)$ are bounded by a positive number $L$. The partial
derivatives of $\rho_{i}(w)$ are computed as follows
\[
\frac{\partial}{\partial x}\rho_{i}(w)=
\frac{x_{0}-x}{[(x-x_{0})^2 + (y-y_{0})^2]^{3/2}},
\]
and
\[
\frac{\partial}{\partial y}\rho_{i}(w)=
\frac{y_{0}-y}{[(x-x_{0})^2 + (y-y_{0})^2]^{3/2}}.
\]
We will show that $|\frac{\partial}{\partial x}\rho_{i}(w)| \leq
L$ (where $L > 0$ is a constant). The proof of the other case,
i.e., $|\frac{\partial}{\partial y}\rho_{i}(w)| \leq L$, is quite
similar.

Since $(x,y)$ is located within a triangle of vertices with finite coordinates, we have $\|x\| < \infty$
and $\|y\| < \infty$. Furthermore, $\sqrt{(x-x_{0})^2 +
(y-y_{0})^2} \geq |x-x_{0}| \geq 0$, and, hence, $[\sqrt{(x-x_{0})^2 +
(y-y_{0})^2}]^{3} \geq |x-x_{0}|^{3}$. This yields:
\[
0 \leq \frac{1}{[\sqrt{(x-x_{0})^2 + (y-y_{0})^2}]^{3}} \leq \frac{1}{|x-x_{0}|^{3}}. \]
Consequently, we have:
\[
0 \leq \frac{|x-x_{0}|}{[\sqrt{(x-x_{0})^2 + (y-y_{0})^2}]^{3}} \leq \frac{1}{|x-x_{0}|^{2}}. \]
The latter inequality reads as follows:
\[
0 \leq | \frac{\partial}{\partial x}\rho_{i}(w) | \leq \frac{1}{|x-x_{0}|^{2}}. \]
Because $|x-x_{0}| \geq 1$ we have $\frac{1}{|x-x_{0}|^{2}} \leq 1$.
Altogether, we showed that $| \frac{\partial}{\partial
x}\rho_{i}(w) | \leq 1$. In a similar way, one can show that $|
\frac{\partial}{\partial y}\rho_{i}(w) | \leq 1$. These
inequalities prove that $\rho_{i}(w)$ is a Lipschitz function.

\textbf{Case (ii): } In this case, the function $\rho_{i}(w)$ is a
constant function; consequently, it is a Lipschitzian.

\textbf{Case (iii): } This case corresponds to
\[
\| w_{x} - g_{i} \| \geq 1 \mbox{  and  }\| w_{y} - g_{i} \| < 1
\]
and we want to show that
\[
\exists L > 0 \quad \mbox{ s.t. } |\rho_{i}(w_{x}) -
\rho_{i}(w_{y})| \leq L \| w_{x} - w_{y} \|.
\]
As $\| w_{y} - g_{i} \| < 1$, Equation~\eqref{eq:rho-lip} yields $\rho_{i}(w_{y})=1$; in addition, we have:
\[
\rho_{i}(w_{x}) = \frac{1}{\sqrt{(x_{1}-x_{0})^2 + (y_{1}-y_{0})^2}}. \]
The triangle inequality yields:
\[
\| w_{x} - w_{y} \| = \| w_{x} - g_{i} + g_{i} - w_{y} \| \geq \|
w_{x} - g_{i} \| - \| g_{i} - w_{y} \|.
\]
Using $\| w_{y} - g_{i} \| < 1$ and $\| w_{x} - g_{i} \| \geq 1$, we obtain:
\[
\| w_{x} - w_{y} \| > \| w_{x} - g_{i} \| - 1 \geq 0.
\]
Because $\| w_{x} - g_{i} \| \geq 1$ holds, we can conclude:
\[
\frac{\| w_{x} - g_{i} \| - 1}{\| w_{x} - g_{i} \|} < \frac{\|
w_{x} - w_{y} \|}{\| w_{x} - g_{i} \|} \leq \| w_{x} - w_{y} \|.
\]
To sum up:
\begin{equation}
\label{eq:Lip-2}
\frac{\| w_{x} - g_{i} \| - 1}{\| w_{x} - g_{i}\|} < \| w_{x} - w_{y} \|. \end{equation}
In addition, we have:
\[
|\rho_{i}(w_{x}) - \rho_{i}(w_{y})| = \left|\frac{1}{\| w_{x} -
g_{i}\|} - 1\right| = \left|\frac{1 - \| w_{x} - g_{i}\|}{\| w_{x} -
g_{i}\|}\right|.
\]
Using $\| w_{x} - g_{i}\| \geq 1$, this yields:
\[
|\rho_{i}(w_{x}) - \rho_{i}(w_{y})| = \frac{\| w_{x} - g_{i}\| -
1}{\| w_{x} - g_{i}\|},
\]
and in combination with~\eqref{eq:Lip-2}, we obtain:
\[
\left|\rho_{i}(w_{x}) - \rho_{i}(w_{y})\right| < \| w_{x} - w_{y} \|.
\]
This shows that $\rho_{i}(\cdot)$ is a Lipschitz function.
Since any finite sum of Lipschitz functions is a Lipschitz
function too, we conclude that the objective function of the
PSP is a Lipschitz function.
\end{proof}

\section{Experiments}
\label{sec:exp}

We evaluate the \Tdiscrete and the \Tcontinuous algorithms from Sections~\ref{sec:appx} and~\ref{sec:lipschitz} experimentally.
The \Tdiscrete algorithm comes in two variants: \Tdiscretecircle and \Tdiscreteoctagon, the former uses circular arcs and the latter uses octagons to approximate the fading function, compare Section~\ref{sec:circles-octagons}.
Regarding the \Tcontinuous algorithm, we focus on the geometric bound from Section~\ref{sec:geom-lb}.
The reason is that computing the geometric bound is simple and efficient, whereas for the Lipschitz bound (Section~\ref{sec:lip-lb}), we cannot achieve comparable guarantees since we do not know a tight value for the Lipschitz constant $L$ in Equation~\eqref{eq:DSP-lower-bound-1}.

In order to evaluate the efficiency and solution quality of the two algorithms from Sections~\ref{sec:appx} and~\ref{sec:lipschitz}, we conducted a series of experiments with implementations for both.
The experimental setup is described in Section~\ref{sec:exp-setup} and the results are analyzed in Section~\ref{sec:exp-results}.

\subsection{Experimental Setup}
\label{sec:exp-setup}

\subsubsection{Instances}

\begin{figure}
	\centering
	\hfill
	\subfigure[\orthoortho polygons are orthogonal polygons with holes.]{
		\includegraphics[width=.4\linewidth]{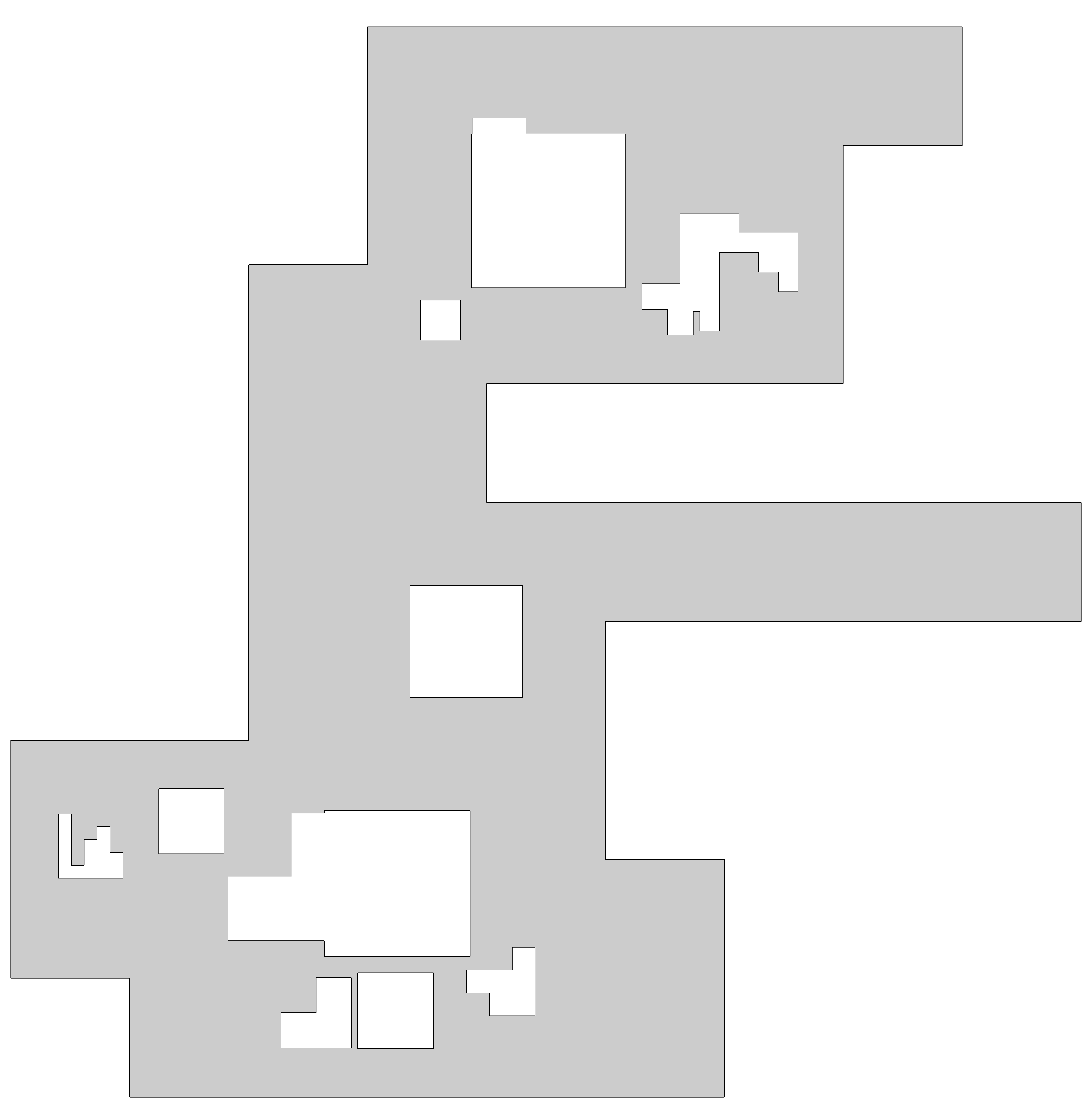}
		\label{fig:instances-ortho-ortho}
	}\hfill
	\subfigure[\spike polygons comprise few positions that see much of the polygon.]{
		\includegraphics[width=.4\linewidth]{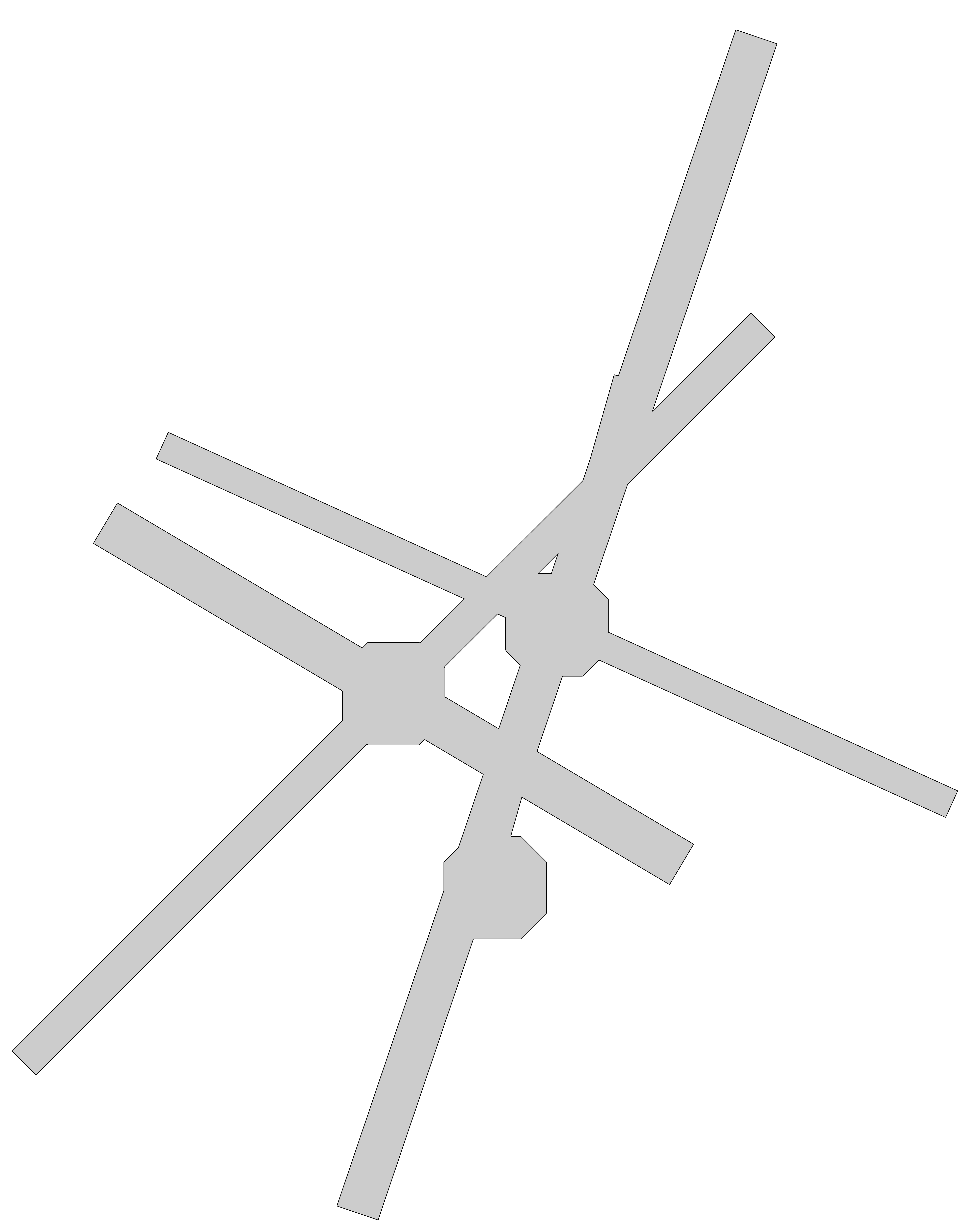}
		\label{fig:instances-spike}
	}\hfill\\
	\hfill
	\subfigure[\vonkoch polygons are inspired by the von Koch curve.]{
		\includegraphics[width=.4\linewidth]{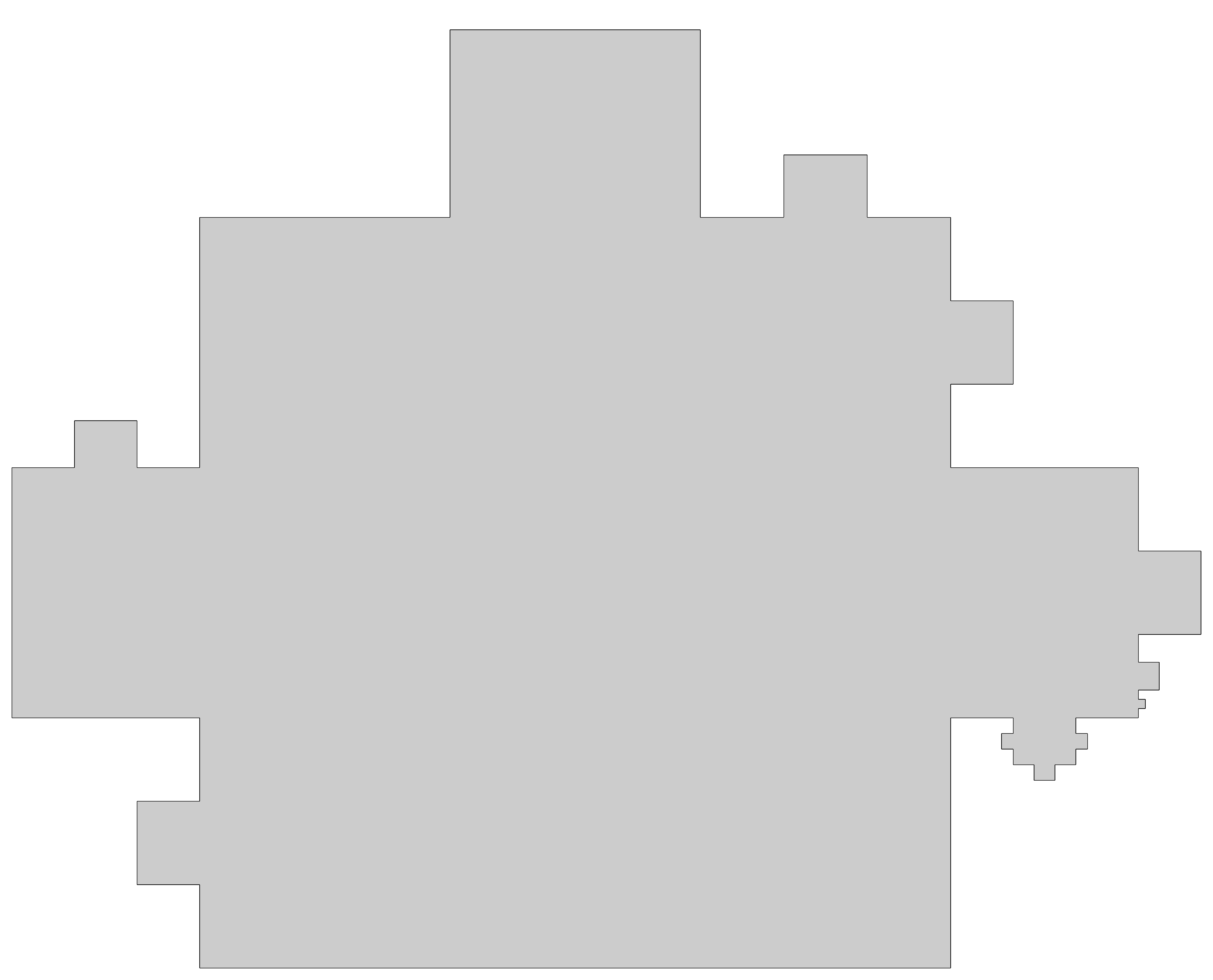}
		\label{fig:instances-vonkoch}
	}\hfill
	\subfigure[\simple polygons are non-orthogonal polygons without holes.]{
		\includegraphics[width=.4\linewidth]{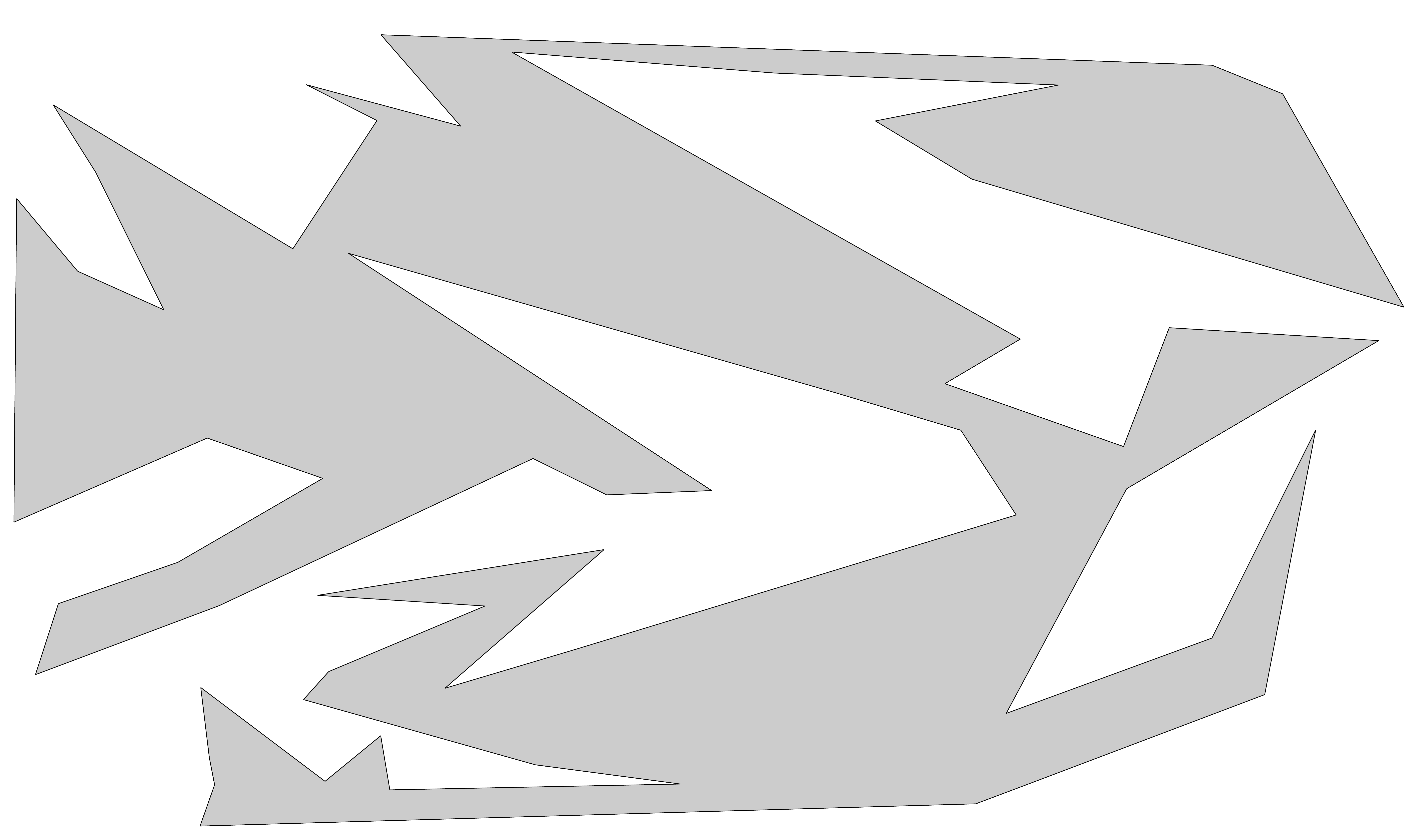}
		\label{fig:instances-simple}
	}\hfill \\
	\subfigure[\simplesimple polygons comprise non-orthogonal polygons with holes.]{
		\includegraphics[width=.8\linewidth]{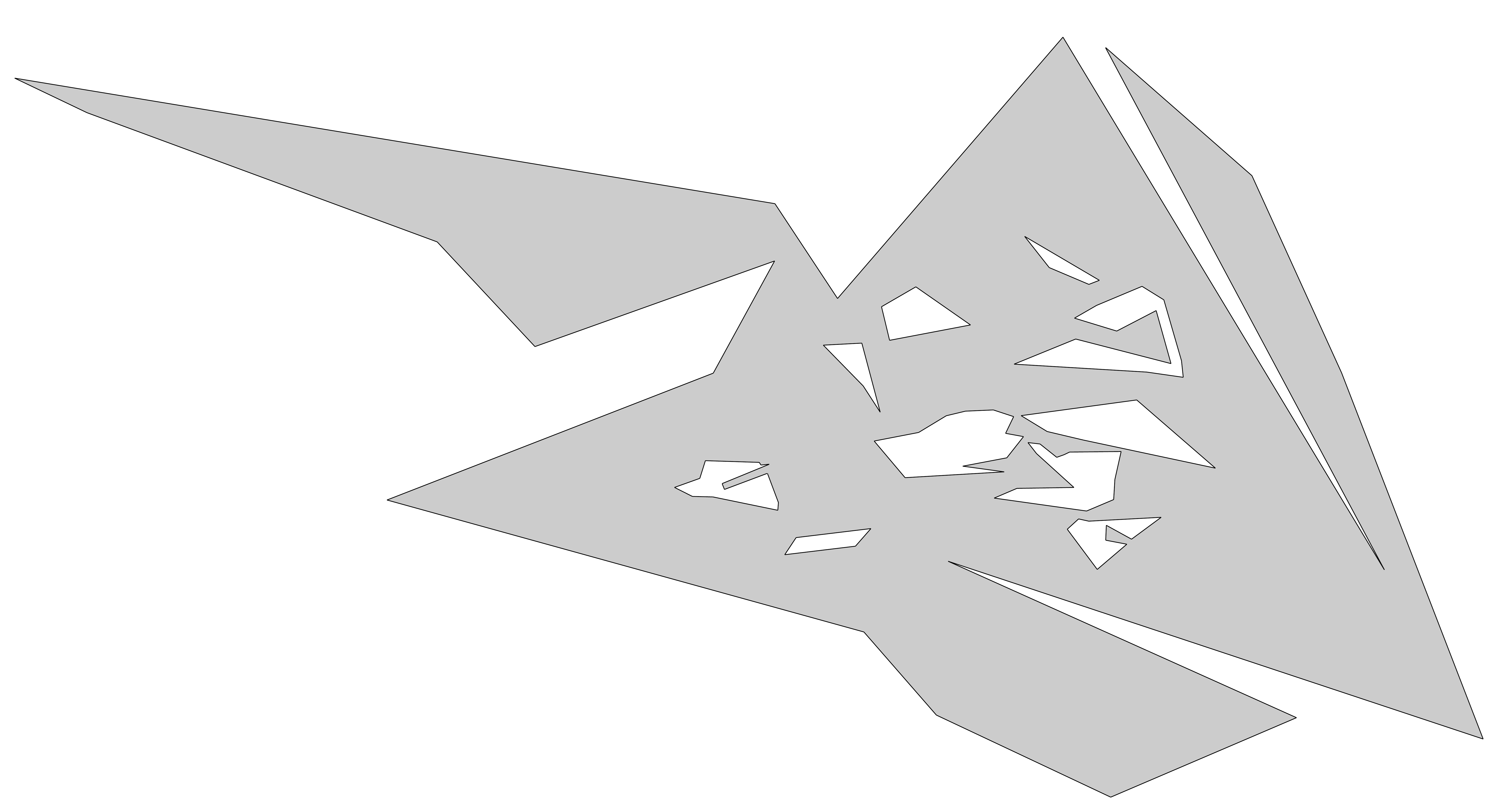}
		\label{fig:instances-simple-simple}
	}
	\caption{		Our test polygon classes.
		Some classes (\orthoortho, \spike, and \simplesimple) have holes, others (\vonkoch and \simple) do not.
		The number of vertices ranges from 50 to~700.}
	\label{fig:instances}
\end{figure}

Our test set is formed by 50 instances from the Art Gallery Problem Instances library~\cite{art-gallery-instances-page}, see Figure~\ref{fig:instances}.
We test five instance classes (\orthoortho, \spike, \vonkoch, \simple, and \simplesimple) with polygons of 50--700 vertices.
As an example for a solved instance, see the \vonkoch-type instance in Figure~\ref{fig:cont-koch}.

\begin{figure}
	\centering
	\subfigure[A \vonkoch polygon with 100 vertices.]{
		\includegraphics[width=.8\linewidth]{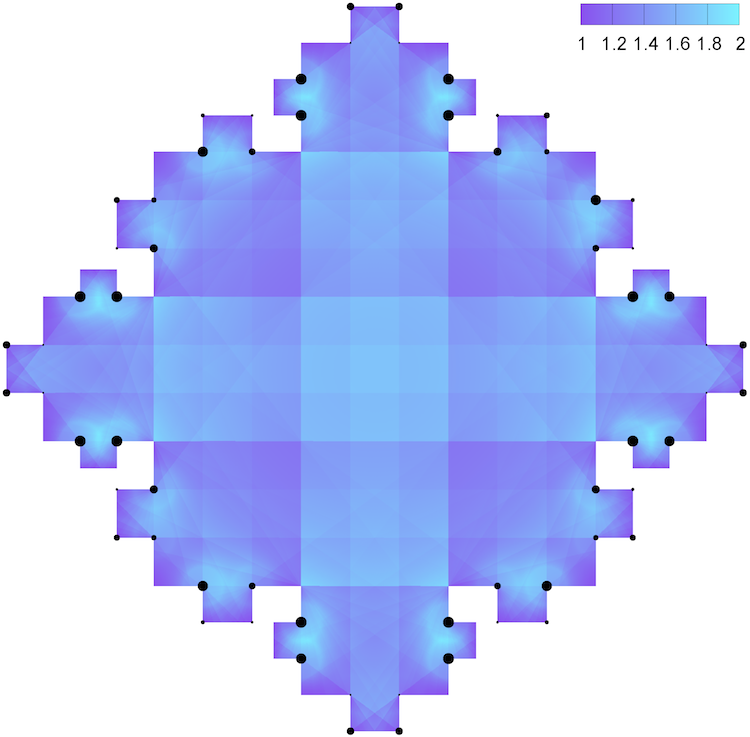}
		\label{fig:cont-koch}
	}
	\subfigure[A 94-vertex office environment.]{
		\includegraphics[width=.8\linewidth]{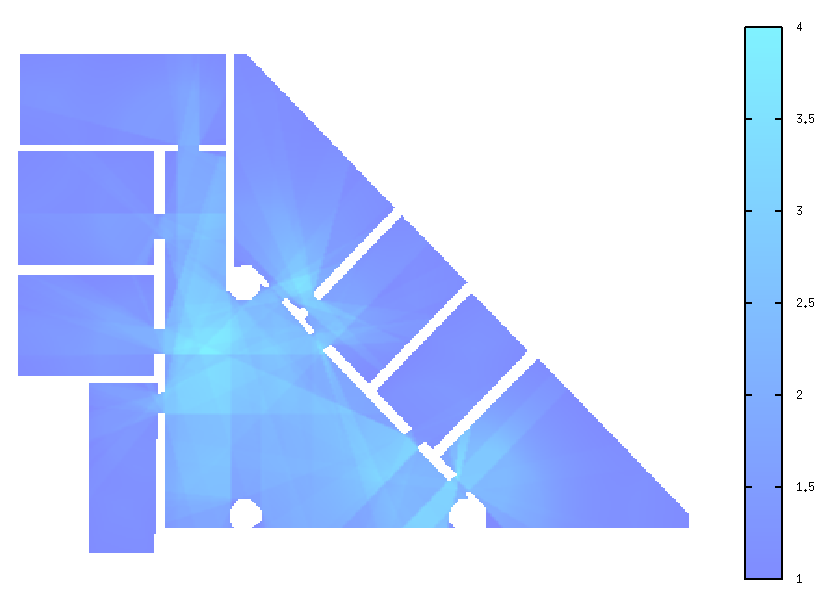}
		\label{fig:cont-office}
	}
	\caption{		AGPF solutions obtained by the \Tcontinuous approach.
		The color scales show values of coverage:
		Dark blue points are lit with exactly the necessary threshold of~1, lighter points obtain more coverage.}
	\label{fig:cont}
\end{figure}

\begin{figure}
	\centering
	\includegraphics[width=.9\linewidth]{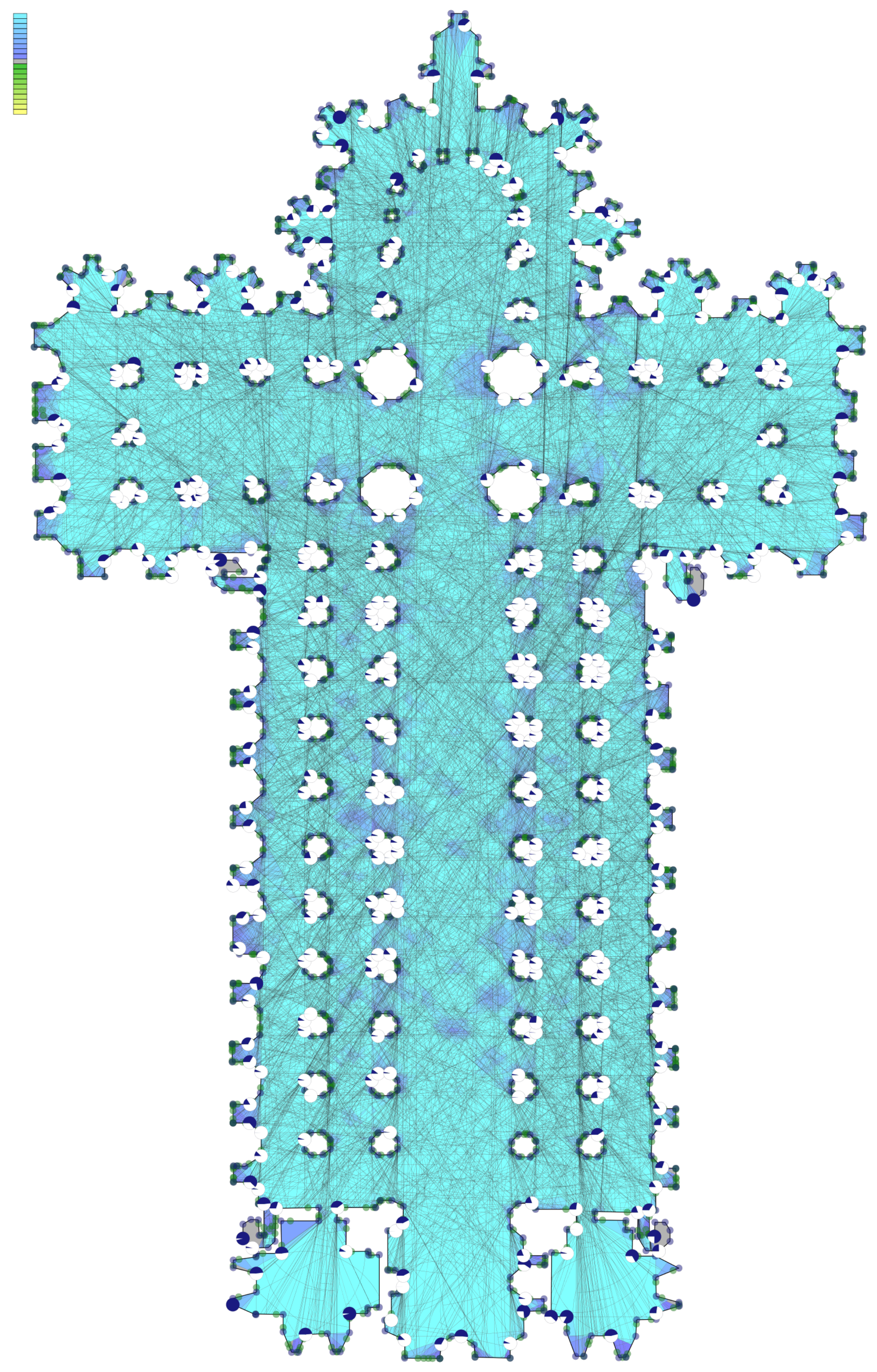}
	\caption{		An AGPF solution of a 1185-vertex cathedral polygon obtained by the \Tdiscretecircle algorithm.
		The color scale in the upper left indicates the value of coverage:
		A gray point is covered with exactly the necessary value of~1, blue indicates that a point obtains light with a value of more than~1 (largest value for light blue);
		yellow and green points are insufficiently lit (not present in this figure).
	}
  \label{fig:disco-cathedral}
\end{figure}

Furthermore, we present a \Tcontinuous solution for a 94-vertex office environment (Figure~\ref{fig:cont-office}), and a \Tdiscretecircle solution for a 1185-vertex cathedral (Figure~\ref{fig:disco-cathedral}).
These instances, however, are a demonstration and as such not included in the regular test set.

\subsubsection{Polygon Scaling}
\label{sec:exp-scaling}

Solutions for the classical AGP are invariant under scaling of the input;
as a result, the input complexity of algorithms solving the AGP is the number of vertices in the polygon.
In the presence of fading, this is no longer the case.
Especially the cutoff of $\rho$ at distance~1, see Equation~\eqref{eq:rho}, becomes an issue:
If every mutually visible pair of points in a polygon is at distance at most~1, the problem becomes equivalent to the AGP without fading.
As an example, consider a convex polygon~$P$:
If $P$ is small enough, one guard placed anywhere in $P$ suffices, but if $P$ is scaled up the structure of an optimal solution becomes non-obvious.
The polygons in the Art Gallery Problem Instance library~\cite{art-gallery-instances-page} have no natural or even consistent spatial size.

The impact of the input polygons' number of vertices on the solution time in the classical AGP, in the form of visibility calculations and arrangement complexity, has been studied extensively~\cite{bhhhk-visibility-14,TR-IC-09-46,crs-ipsagp-09-XXXX-STATTDESSEN-DER-TR-XXX,crs-exmvg-11,csr-eeeaoagp-08,csr-eeaoagp-07,rsfhkt-eag-14,ffks-ffagp-14,kbfs-esbgagp-12,DaviPedroCid-J-002013}.
Given that fading introduces spatial size as a new source of complexity, we confine our experiments to polygons with up to 700 vertices and focus on the impact of spatial extent.

We scale the input polygons $P$ as follows:
First, we compute an estimate for the spatial complexity of~$P$.
Here, different measures could be used, for example the diameter of~$P$, the area of~$P$, or local feature size.
Pre-experiments have shown that the average edge length of~$P$, denoted by~$\mu(P)$, provides an intuitively consistent measure for our test set.
Hence, we scale $P$ by a factor $(\Lambda\mu(P))^{-1}$, for a parameter~$\Lambda$.
Therefore, $\Lambda = 2$ corresponds to having the cutoff in $\rho$ at twice the average edge length of~$P$.
We scaled our 50 input polygons using 0.2, 0.5, 1, and~2 as possible values for~$\Lambda$, yielding 200 scaled instances.

\subsubsection{Parameterization}

The algorithms are parameterized.
Both algorithms account for the fading exponent~$\alpha$;
we selected the two practically relevant values of~1 (laser scanner) and~2 (light).

The \Tdiscrete algorithm takes the desired approximation factor $1 + \epsilon$ as an input, we selected 1.2, 1.6, and~2.
Furthermore, we test the \Tdiscrete algorithm using circles as well as octagons, \Tdiscretecircle and \Tdiscreteoctagon, to approximate the fading function; this is discussed in Section~\ref{sec:circles-octagons}.

The \Tcontinuous algorithm is parameterized in its stopping criterion:
A triangle with a difference of $\delta$ between lower and upper bound is no longer partitioned.
We pick 0.01, 0.001, and 0.0001 as values for~$\delta$.

In total, this yields nine configurations, six \Tdiscrete and three \Tcontinuous, each with two choices for~$\alpha$.
Together with the 200 scaled inputs, we run 3600 configuration-instance pairs.

\subsubsection{Hardware and Implementation}
\label{sec:exp-hardware}

Our implementation of the algorithms from Sections~\ref{sec:appx} and~\ref{sec:lipschitz} is based on our implementation for solving the classical AGP~\cite{ffks-ffagp-14,kbfs-esbgagp-12}.
We used CPLEX-12.6~\cite{cplex} to solve LPs and CGAL-4.4~\cite{cgal,cgal:cpt-cgk2-15a,hhb-visibility-2-15b} for geometric subroutines.
All experiments were carried out on a Linux-3.18.27 server with two Intel Xeon E5-2680 CPUs, i.e., 16 cores and 32 threads, and 258\,GiB of RAM.
We ran 15 instances in parallel, each limited to 20 minutes of CPU time and 16\,GB of RAM, aborting every run which did not finish within these limits.
Our implementation is not parallelized.

\subsection{Results}
\label{sec:exp-results}

Tables~\ref{tab:exp-rate}--\ref{tab:exp-cpu} depict the results of our experiments.
For each algorithm configuration, Table~\ref{tab:exp-rate} addresses its success rate, Table~\ref{tab:exp-obj} its objective values, and Table~\ref{tab:exp-cpu} the required CPU time.
We discuss these aspects in Sections~\ref{sec:exp-rate}, \ref{sec:exp-quality}, and~\ref{sec:exp-efficiency}, respectively.
Note that each table cell corresponds to 50 test runs, i.e., one per instance.

\subsubsection{Success Rate}
\label{sec:exp-rate}

\begin{table}
	\centering
	\caption{Success rates in percent.}
	\label{tab:exp-rate}
	\begin{tabular}{|l|cccc|cccc|}\hline
		\Talpha&\multicolumn{4}{c|}{1}&\multicolumn{4}{c|}{2}\\
		\Tradius&2&1&0.5&0.2&2&1&0.5&0.2\\\hline
		\Tcontinuous \hfill 0.01 & 98 & 98 & 94 & 92 & 98 & 98 & 90 & 88\\
		\Tcontinuous \hfill 0.001 & 98 & 96 & 88 & 86 & 96 & 88 & 84 & 66\\
		\Tcontinuous \hfill 0.0001 & 98 & 88 & 84 & 68 & 96 & 86 & 72 & 42\\
		\hline
		\Tdiscretecircle  \hfill 2 & 100 & 100 & 98 & 94 & 100 & 98 & 88 & 80\\
		\Tdiscretecircle  \hfill 1.6 & 100 & 100 & 98 & 84 & 100 & 98 & 84 & 60\\
		\Tdiscretecircle  \hfill 1.2 & 100 & 98 & 88 & 60 & 98 & 92 & 58 & 40\\
		\hline
		\Tdiscreteoctagon \hfill 2 & 100 & 100 & 96 & 82 & 100 & 98 & 90 & 72\\
		\Tdiscreteoctagon \hfill 1.6 & 100 & 100 & 96 & 76 & 100 & 98 & 84 & 60\\
		\Tdiscreteoctagon \hfill 1.2 & 98 & 96 & 78 & 48 & 84 & 32 & 10 & 4\\
		\hline
	\end{tabular}
\end{table}

We report the success rate of each algorithm configuration in Table~\ref{tab:exp-rate} w.r.t.\ the imposed CPU time and memory limits described in Section~\ref{sec:exp-hardware}.
Out of the 3600 test runs, 84.9\,\% succeeded, 12.3\,\% did not meet the CPU-time limit of 20 minutes, and only 2.9\,\% exceeded the 16\,GB memory limit (compare Section~\ref{sec:exp-hardware}).

In the following, we account for the times and objective values of failed runs with infinity and report median values.

\subsubsection{Objective Value}
\label{sec:exp-quality}

\begin{table}
	\centering
	\caption{Median of relative objective values.}
	\label{tab:exp-obj}
	\begin{tabular}{|l|cccc|cccc|}\hline
		\Talpha&\multicolumn{4}{c|}{1}&\multicolumn{4}{c|}{2}\\
		\Tradius&2&1&0.5&0.2&2&1&0.5&0.2\\\hline
          Lower Bound& 0.85 & 0.88 & 0.90 & 0.93 & 0.85 & 0.88 & 0.91 & 0.87 \\ \hline
		\Tcontinuous \hfill 0.01 & 1.00 & 1.00 & 1.00 & 1.00 & 1.00 & 1.00 & 1.00 & 1.00\\
		\Tcontinuous \hfill 0.001 & 1.00 & 1.00 & 1.00 & 1.00 & 1.00 & 1.00 & 1.00 & 1.00\\
		\Tcontinuous \hfill 0.0001 & 1.00 & 1.00 & 1.00 & 1.00 & 1.00 & 1.00 & 1.00 & $\infty$\\
		\hline
		\Tdiscretecircle \hfill 2 & 1.10 & 1.22 & 1.37 & 1.52 & 1.10 & 1.31 & 1.52 & 1.59\\
		\Tdiscretecircle \hfill 1.6 & 1.07 & 1.14 & 1.24 & 1.34 & 1.06 & 1.19 & 1.33 & 1.41\\
		\Tdiscretecircle \hfill 1.2 & 1.02 & 1.05 & 1.09 & 1.13 & 1.02 & 1.06 & 1.14 & $\infty$\\
		\hline
		\Tdiscreteoctagon \hfill 2 & 1.11 & 1.23 & 1.38 & 1.53 & 1.09 & 1.30 & 1.49 & 1.60\\
		\Tdiscreteoctagon \hfill 1.6 & 1.07 & 1.15 & 1.25 & 1.35 & 1.07 & 1.17 & 1.32 & 1.40\\
		\Tdiscreteoctagon \hfill 1.2 & 1.02 & 1.06 & 1.09 & $\infty$ & 1.02 & $\infty$ & $\infty$ & $\infty$\\
		\hline
	\end{tabular}
\end{table}

Table~\ref{tab:exp-obj} compares the objective values determined by the different algorithm configurations.
Per instance, we consider the objective value relative to the best known solution (w.r.t.\ the same choice for $\alpha$ and~$\Lambda$) and report, per cell, their median.
For example, the 1.02 in the lower left corner encodes that, with $\alpha = 1$ and $\Lambda = 2$, the \Tdiscreteoctagon algorithm with $1 + \epsilon = 1.2$ produced solutions at 1.02 times the objective value of the best algorithm (\Tcontinuous with $\delta = 0.01$).

Table~\ref{tab:exp-obj} indicates that, where applicable, the median solution quality of the \Tcontinuous algorithm does not vary w.r.t.\ the choice of~$\delta$.
Observe, however, that this does not mean that the solutions are optimal:
The \Tcontinuous algorithm does not guarantee an approximation ratio.

For guaranteed bounds we include the row ``Lower Bound.'' Each \Tdiscrete algorithm, configured as $(1+\epsilon)$-approximation, terminating with an objective value of $z$ yields a lower bound of $(1+\epsilon)^{-1} z$. We present the median of the largest lower bound per instance (if applicable).
Even with these conservative bounds, it is clear that the \Tcontinuous algorithm consistently produces good solutions.

The \Tdiscrete algorithm (\Tdiscreteoctagon as well as \Tdiscretecircle) usually produces more expensive solutions than \Tcontinuous, but comes with the benefit of a proven approximation factor.
Assuming the \Tcontinuous solutions to be near-optimal, it can be seen that the actual optimality gap (i.e., the gap between \Tdiscrete and \Tcontinuous) improves with larger values for~$\Lambda$, i.e., when scaling down polygons.
This might be due to the fact that more space is covered with 1-disks around guards, where $\rho$ and $\tau$ are equal.

\subsubsection{Efficiency}
\label{sec:exp-efficiency}

\begin{table}
	\centering
	\caption{Median of relative CPU times.}
	\label{tab:exp-cpu}
	\begin{tabular}{|l|cccc|cccc|}\hline
		\Talpha&\multicolumn{4}{c|}{1}&\multicolumn{4}{c|}{2}\\
		\Tradius&2&1&0.5&0.2&2&1&0.5&0.2\\\hline
		\Tcontinuous \hfill 0.01 & 1.89 & 2.45 & 4.53 & 8.87 & 2.45 & 3.36 & 8.55 & 19.49\\
		\Tcontinuous \hfill 0.001 & 1.98 & 3.20 & 7.91 & 19.16 & 2.37 & 5.17 & 18.10 & 173.97\\
		\Tcontinuous \hfill 0.0001 & 1.94 & 4.63 & 14.02 & 110.47 & 2.68 & 7.90 & 74.81 & $\infty$\\
		\hline
		\Tdiscretecircle \hfill 2 & 1.00 & 1.91 & 3.70 & 6.97 & 1.15 & 2.82 & 7.58 & 17.48\\
		\Tdiscretecircle \hfill 1.6 & 1.05 & 2.07 & 4.62 & 9.88 & 1.23 & 3.44 & 11.23 & 36.96\\
		\Tdiscretecircle \hfill 1.2 & 1.22 & 3.12 & 9.70 & 39.80 & 1.67 & 7.33 & 72.96 & $\infty$\\
		\hline
		\Tdiscreteoctagon \hfill 2 & 1.13 & 1.91 & 3.45 & 6.59 & 1.38 & 3.08 & 8.26 & 22.27\\
		\Tdiscreteoctagon \hfill 1.6 & 1.14 & 2.19 & 4.36 & 11.59 & 1.62 & 4.70 & 16.62 & 65.52\\
		\Tdiscreteoctagon \hfill 1.2 & 1.80 & 4.76 & 20.69 & $\infty$ & 13.62 & $\infty$ & $\infty$ & $\infty$\\
		\hline
	\end{tabular}
\end{table}

We compare the CPU time of the different algorithm configurations in Table~\ref{tab:exp-cpu}.
Per instance, we normalize the run time to that of \Tdiscretecircle with $1+\epsilon = 2$ for $\alpha = 1$ and $\Lambda = 2$, the fastest configuration.
In each cell we report the median relative run time over all instances.
For example, the 1.80 in the lower left corner encodes that, with $\alpha = 1$ and $\Lambda = 2$, the \Tdiscreteoctagon algorithm with $1 + \epsilon = 1.2$ took 1.8 times as long as the \Tdiscretecircle configuration specified above.

For small~$\Lambda$, i.e., scaled-up polygons, it can be seen that the choice of~$\delta$ has a substantial impact on the CPU time of \Tcontinuous. It has, however, no measurable impact on the solution quality as discussed in Section~\ref{sec:exp-quality}.

Increasing $\Lambda$ narrows the speed gap between \Tdiscrete and \Tcontinuous, as the former profits from a large $\Lambda$ by using fewer steps in~$\tau$\dash---and hence fewer arcs in the arrangements\dash---whereas it makes little difference to \Tcontinuous.

Furthermore, in most configurations, \Tdiscretecircle is faster than \Tdiscreteoctagon, especially for $\alpha = 2$.
This means that the issue discussed in Section~\ref{sec:circles-octagons} is settled:
Using octagons instead of circles in order to exclude circular arcs from the arrangements does not pay off, especially since it takes more octagons than circles to achieve the same approximation factor w.r.t.\ the continuous fading function~$\rho$.

\subsection{Choosing an Algorithm}

Clearly, \Tdiscretecircle should be preferred over \Tdiscreteoctagon as it has higher success rates, smaller objective values, and, in most configurations, lower run times.

As for \Tcontinuous, observe that in our experiments the choice of $\delta$ does not influence the objective value.
Hence, we compare \Tdiscretecircle to \Tcontinuous with $\delta = 0.01$.
The only configuration of \Tdiscretecircle yielding comparable objective values uses $1 + \epsilon = 1.2$, but even in that case \Tcontinuous usually produces better solutions and has higher success rates.

However, \Tcontinuous does not guarantee any bounds on the approximation quality.
Hence, if an approximation guarantee is required, \Tdiscretecircle must be preferred.
Furthermore, \Tdiscretecircle is faster for large~$\Lambda$, i.e., on polygons with a small spatial extent.
 
\section{Future Work: AGPF with Point Guards}\label{sec:point}

Sections~\ref{sec:fixedguards} and~\ref{sec:exp} demonstrate that $\AGC(G,P)$ can be solved efficiently with the proposed algorithms.
Hence, solving $\AGC(P,P)$ is a natural candidate for future work.

In previous work~\cite{kbfs-esbgagp-12}, we successfully used the doubly-infinite LP formulation for $\AG(P,P)$, together with separation routines for both the primal and the dual~LP, in a practically efficient algorithm.
The same could be applied here.
As can be easily confirmed, the dual separation problem asks for a point $g \in P$ with
\begin{equation}
  \sum_{w\in\vis{g}} \varrho(g,w)y_w > 1,
\end{equation}
given a dual solution $y$.
Primal and dual separation problems are very similar:
Given a solution~$x$ ($y$, respectively), they ask for the darkest (respectively brightest) point in the corresponding arrangement.

If those two problems can be solved efficiently, a possible strategy is to switch between primal and dual separation.
This strategy is not guaranteed to converge, but if it does it converges to an optimal solution\dash---and it has been shown to perform well in practice~\cite{kbfs-esbgagp-12}.
Thus, the interesting question is whether this holds true for $\AGC(P,P)$.
Both approaches from this paper allow for such a strategy.

Another possibility is to use the approach by Tozoni et~al.~\cite{DaviPedroCid-J-002013} for the AGP with point guards, which uses inclusion-maximal features of a witness visibility-arrangement for guard placement.
This generalizes to $\AGC_\tau$:
A feature of the witness overlay dominates another if it is seen by at least the same set of guards and their coefficients $\varrho(g,w)$ are not smaller.

\section{Conclusions}
\label{sec:concl}

We introduce the Art Gallery Problem with Fading~(AGPF), which is a generalization of both, the well-established Art Gallery Problem~\cite{rsfhkt-eag-14,r-agta-87} and the Stage Illumination Problem by Eisenbrand et~al.~\cite{efkm-esi-05}.
For the AGPF, we present two efficient algorithms for the case with fixed guard positions stemming from an infinite LP formulation.
While one of them is faster in practice and generally provides a good solution quality,
the other provides a fully polynomial-time approximation scheme, i.e., guaranteed bounds on the solution quality.
We evaluate both approaches experimentally.
 
\section*{Acknowledgments}
We would like to thank the anonymous reviewers for their suggestions that helped improve our presentation.

This work was partially supported by the Deutsche For\-schungs\-gemein\-schaft (DFG) under contract number KR~3133/1-1 (Kunst!).
Stephan Friedrichs, Mahdi Moeini (through Kunst!), and Christiane Schmidt were affiliated with the Technical University of Braunschweig during part of the research.
Mahdi Moeini was supported by both CNRS and OSEO within the ISI project ``Pajero'' (France).
Currently, Mahdi Moeini is affiliated to the Technical University of Kaiserslautern (Germany) through the program ``CoVaCo''.
Christiane Schmidt is supported by grant 2014-03476 from Sweden's innovation agency VINNOVA, and was supported by the Israeli Centers of Research Excellence (I-CORE) program (Center No. 4/11) during the development of the paper.

\bibliographystyle{abbrv}
\bibliography{refs2}

\end{document}